\pdfoutput=1
\documentclass[journal,twoside,web]{ieeecolor}
\usepackage{generic}
\usepackage{cite}
\usepackage{amsmath,amssymb,amsfonts}
\usepackage{algorithm}
\usepackage{algpseudocode}
\usepackage{graphicx}
\usepackage{epstopdf}
\usepackage{optidef}
\usepackage{textcomp}
\usepackage{csquotes}
\usepackage{url}
\usepackage{multicol}
\usepackage{color}
\usepackage{tabularray}
\newtheorem{defi}{Definition}

\newtheorem{theorem}{Theorem}
\newtheorem{crl}{Corollary}

\newtheorem{proposition}{Proposition}
\newtheorem{assumption}{Assumption}
\allowdisplaybreaks

\def\BibTeX{{\rm B\kern-.05em{\sc i\kern-.025em b}\kern-.08em
    T\kern-.1667em\lower.7ex\hbox{E}\kern-.125emX}}
\begin{document}
\title{Multi-Step Optimal Tracking Control of \\ Unknown Nonzero-Sum Games based on\\ Least Squares and Linear Programming:\\ An Application to a Fully-Automated, Dual-Hormone Artificial Pancreas}
\author{Alexandros Tanzanakis, \IEEEmembership{Member, IEEE}, and John Lygeros, \IEEEmembership{Fellow, IEEE}
        % <-this % stops a space
\thanks{Alexandros Tanzanakis and John Lygeros: Department of Information Technology and Electrical Engineering, ETH Zurich, Switzerland, {\tt\small\{atanzana,jlygeros\}@ethz.ch}.}
\thanks{This research work was supported by the European Research Council (ERC) under the project OCAL, grant number 787845.}
}

\maketitle

\begin{abstract}
We consider the problem of optimal tracking control of unknown discrete-time nonlinear nonzero-sum games. The related state-of-art literature is mostly focused on Policy Iteration algorithms and multiple neural network approximation, which may lead to practical implementation challenges and high computational burden. To overcome these problems, we propose a novel Q-function-based multi-step Value Iteration algorithm, which provides the potential to accelerate convergence speed and improve the quality of solutions, with an easy-to-realize initialization condition. A critic-only least squares implementation approach is then employed, which alleviates the computational complexity of commonly used multiple neural network-based methods. Afterwards, by introducing the coupled Bellman operator, a novel linear programming approach is derived, based on which Nash equilibria can be approximately computed by solving a set of tractable finite-dimensional optimization problems. We evaluate the tracking control capabilities of the proposed algorithms to the problem of fully-automated dual-hormone (i.e., insulin and glucagon) glucose control in Type 1 Diabetes Mellitus. The U.S. FDA-accepted DMMS.R simulator from the Epsilon Group is used to conduct extensive in-silico clinical studies on virtual patients under a variety of completely unannounced meal and exercise scenarios. Simulation results demonstrate the high reliability and exceptional performance of the proposed multi-step algorithmic framework to critical complex systems.
\end{abstract}

\begin{IEEEkeywords}
approximate dynamic programming, artificial pancreas, data-driven control, diabetes mellitus, nonzero-sum games, reinforcement learning.
\end{IEEEkeywords}

\section{Introduction}
\label{sec:introduction}
\IEEEPARstart{M}{any} realistic complex systems involve more than one control input \cite{b1,b2,b3}. By treating control inputs as strategies employed by different players, game theoretical control methods can be utilized by formulating a multiplayer game \cite{b1,b4}. Under this setting, players employ concurrent strategies (i.e., control inputs) to optimize coupled cost functions, until they converge to a Nash equilibrium (NE) for which no player can achieve a better performance outcome by individually modifying its own strategy. In recent years, nonzero-sum games (NZSGs) \cite{b5} have been receiving a great deal of attention in the learning-based control community. In contrast to fully-cooperative \cite{b6} and zero-sum games \cite{b7}, NZSGs assume that neither player is fully cooperative nor fully competitive. NE solutions for multiplayer NZSGs can be theoretically obtained by solving a set of coupled Bellman equations \cite{b1,b4,b8}. However, it is generally very difficult or impossible to compute analytic NE solutions due to the nonlinearity and coupling of the associated equations. To overcome these challenges, approximate dynamic programming (ADP) \cite{b1,b8}, merging reinforcement learning (RL) control methods with function approximation, is widely employed to enable approximate optimal control.\\
Numerous ADP algorithms have been proposed to compute approximate NE solutions for continuous-time and discrete-time NZSGs and related graphical games. The vast majority in state-of-art literature is focused on Policy Iteration (PI) algorithms \cite{b9,b10,b11,b12,b13,b14,b15,b16,b17,b18,b19,b20,b21,b22}. PI provides relatively fast convergence to approximate optimal solutions, but requires an initial set of stabilizing control policies; this can lead to critical implementation challenges for complex systems. A few research works have considered Value Iteration (VI) algorithms \cite{b23,b24,b25,b26,b27}; they provide theoretical monotonicity and convergence guarantees but only for NZSGs and graphical games with linear dynamics. For the classical case of single-input systems, VI is well-known to benefit from relaxed initialization conditions, at the price of achieving slower convergence compared to PI \cite{b28,b29,b30}.\\
Motivated by PI and VI, the derivation of novel algorithms which attempt to unify the merits of both methods, has recently attracted attention. Towards this direction, multi-step RL \cite{b31} introduces a limited lookahead data horizon in the policy evaluation and/or policy improvement stages of a respective RL algorithm. Various multi-step PI and VI algorithms have been proposed for discounted finite Markov Decision Processes (MDPs) \cite{b32,b33,b34,b35,b36,b37,b38} and single-input dynamical systems \cite{b39,b40}, which address the positive impact of exploiting multi-step trajectory data to convergence speed. However, for the derived multi-step PI algorithms, the requirement for an initial stabilizing control policy still holds. Furthermore, while multi-step PI has been extended to discrete-time NZSGs \cite{b41}, this is not the case for multi-step VI, remaining an important open problem.\\
In contrast to least squares (LS) and neural network (NN)-based optimal control approaches widely employed in the ADP literature, the linear programmming (LP) approach \cite{b42,b43,b44} exploits the theoretically proven monotonicity and contractivity properties of the discrete-time Bellman operator to construct an infinite-dimensional optimization problem whose solution concurs with the optimal solution of the discrete-time Bellman equation. Due to potential intractability issues, approximation methods are employed to formulate easy-to-solve, tractable finite-dimensional linear programs \cite{b48}. The LP approach has received interest in both model-based and model-free control, although currently only for single-input deterministic and stochastic systems \cite{b40,b45,b46,b47,b48,b49,b50}. Therefore, it is of major importance to investigate and extend this approach to multiplayer game settings.\\
Learning-based control methods have shown great potential to be successfuly applied on problems related to personalized medicine of chronic diseases, e.g., on Diabetes Mellitus (DM) \cite{b51,b52}. The pancreatic beta cells, which are responsible for producing a glucose-decreasing hormone known as insulin, are permanently destroyed by the immune system in Type 1 DM (T1DM) \cite{b53,b54}, a serious chronic disease. T1DM patients suffer from critical hyperglycemia, a condition that has extremely serious long-term effects. As a result, they are required to receive life-long exogenous insulin \cite{b55}. To this end, the artificial pancreas (AP) has emerged as the most cutting-edge closed-loop T1DM treatment option \cite{b56}. In such a system, a control algorithm closes the loop between a continuous glucose monitor (CGM) and an insulin pump by computing the exact amount of insulin that the pump should administer in response to CGM measurements.\\
The vast majority of currently available AP systems use only insulin as a control input. Despite the relevant success of these single-hormone (SH) AP systems in basic lifestyle scenarios, there are various challenges to be tackled \cite{b57,b58,b59,b60}.  In particular, a T1DM patient may experience severe insulin induced hypoglycaemia due to patient-associated metabolic delays (e.g., affecting insulin action) and gradual dysfunction of pancreatic alpha cells, leading to severely impaired secretion of the glucose-elevating hormone known as glucagon. Therefore, one promising approach is the introduction of exogenous glucagon infusion as an additional control input to the process. This enables the design of the so-called dual-hormone (DH) AP systems. The related state-of-art literature, despite showing promising results in reducing hypoglycaemia in comparison to SH AP systems, is currently at a relatively early stage of algorithmic and medical development \cite{b60,b61,b62,b63,b64,b65,b66,b67,b68,b69,b70,b71}. One of the most critical open challenges up-to-date is the design of a fully-automated, DH AP system, where a control algorithm automatically computes clinically safe doses of insulin and glucagon to be administered to the patient based only on closed-loop CGM measurements and without any intervention from the patient (e.g., in the form of meal and exercise announcements to the AP system). This would enable a significantly more convenient lifestyle for all T1DM patients worldwide.\\
The contributions of this work are summarized as follows:
\begin{itemize}
\item[1)] We derive a novel Q-function-based multi-step VI (MSQVI) algorithm for optimal tracking control of completely unknown discrete-time deterministic nonlinear NZSGs. The proposed algorithm benefits from strong monotonicity and convergence guarantees, while it can potentially achieve higher convergence speed and improved NE solutions (in terms of tracking control performance) compared to standard VI. It also enjoys an easy-to-realize initialization condition.
\item[2)] A critic-only LS approach with linear function approximation is employed to implement the proposed MSQVI algorithm, which sharply reduces the computational burden compared to commonly used multiple NN approximation methods \cite{b8,b9,b10,b11,b12,b13,b14,b15,b16,b17,b18,b19,b20,b21,b22,b23,b24,b25,b26,b27}.
\item[3)] We introduce a discrete-time coupled Bellman operator, which inherits the monotone contraction property of the standard discrete-time Bellman operator. Based on the derived operator, we show how to construct a set of infinite-dimensional LPs, whose solutions constitute NE since they coincide with the optimal solutions to the discrete-time coupled Bellman optimality equations. To tackle possible intractability issues, we employ a critic-only, data-driven approximation method which leads to the derivation of tractable, finite-dimensional LPs. This enables the derivation of a LP-based MSQVI formulation.
\item[4)] We evaluate the performance and suitability of the proposed model-free MSQVI framework to personalized drug delivery systems, in particular the fully-automated, dual-hormone (i.e., insulin and glucagon) glucose control of patients diagnosed with T1DM. A clinically validated metabolic simulator is used to conduct extensive in-silico clinical studies on representantive virtual adult subjects, under a variety of completely unannounced meal and exercise scenarios.
\end{itemize}

The structure of the paper is given as follows. The problem definition is given in Section II. The proposed MSQVI algorithm is presented and analyzed in Section III. A data-driven critic-only implementation approach based on LS is proposed in Section IV. The proposed LP approach for discrete-time NZSGs is derived in Section V, and the conducted in-silico clinical studies are presented and discussed in Section VI. Finally, conclusions are given in Section VII.\\
\textbf{Notation.} $\mathbb{N}$ and $\mathbb{N}_{0}$ are the sets of natural numbers and natural numbers including $0$ respectively. $\mathbb{R_{+}}$ and $\mathbb{R_{++}}$ refer to the sets of non-negative and positive real numbers respectively. $I_{n}$ defines an identity matrix of size $n\times n$. $\mathbb{S}^{n}_{++}$ defines the set of symmetric positive definite matrices of size $n\times n$. $\overline{\sigma}(A)$ and $\underline{\sigma}(A)$ define the maximum and minimum singular value of a matrix $A\in \mathbb{R}^{m\times n}$ respectively. \\ \textit{Regarding measurement units:} mg defines milligrams, while mg/dL defines milligrams per deciliter and mg/5mins refers to milligrams per $5$ minutes. U defines units, while U/5mins defines units per $5$ minutes.  

\section{Problem Statement}
We study the following class of discrete-time nonlinear systems with $N$ control inputs (or players) given by
\begin{equation}
\label{eq:1}
x_{k+1} = f(x_k)+\sum_{i=1}^{N}g_{i}(x_k)u_{ik},
\end{equation}
where $x\in \mathcal{X}\subset \mathbb{R}^{n}$ and $u_i \in \mathcal{U}_i\subset \mathbb{R}^{m_i}$ denote the system state and control input applied by player $i\in \mathcal{P}\triangleq\{1,2,\ldots,N\}$ at time step $k\in \mathbb{N}_0$, while $f(x):\mathcal{X} \rightarrow \mathcal{X}$ and $g_{i}(x):\mathcal{X}\rightarrow \mathcal{Y}_i\subset \mathbb{R}^{n\times m_i}$ define the drift and control input dynamics respectively. Furthermore, let $r\in \mathcal{X}$ be a bounded reference input with dynamics defined by an exosystem
\begin{equation}
\label{eq:2}
r_{k+1} = h(r_k),
\end{equation}
with $h:\mathcal{X}\rightarrow \mathcal{X}$. The variables $(x,r)$ are centralized and available to all players $i\in \mathcal{P}$. The main objective is the computation of a set of state-reference feedback control policies $\{\mu_{i}(x,r)\}_{i=1}^{N}$ with $\mu_i:\mathcal{X}^{2}\rightarrow \mathcal{U}_i$, such that the closed-loop nonlinear system \eqref{eq:1} tracks a desired reference signal \eqref{eq:2}.
\begin{assumption}
\label{ass:1}
$\mathcal{X}$, $\{\mathcal{U}_i\}_{i=1}^{N}$ and $\{\mathcal{Y}_i\}_{i=1}^{N}$ are compact sets which contain the origin. The functions $f(x)$ and $\{g_i(x)\}_{i=1}^{N}$ are continuously differentiable on $\mathcal{X}$ with $f(0)=g_i(0)=0$. The function $h(r)$ is Lipschitz continuous on $\mathcal{X}$ with $h(0)=0$. The functional forms of $f(x)$, $\{g_i(x)\}_{i=1}^{N}$ and $h(r)$ are unknown. The system \eqref{eq:1} is controllable on $\mathcal{X}$.
\end{assumption}
By defining $-i \triangleq \{j: j\in \mathcal{P}, j\neq i\}$ and $\mathcal{Z}_i= \mathcal{X}^{2}\times \mathcal{U}_i\times \mathcal{U}_{-i}$ for all $i\in \mathcal{P}$, each player $i$ contributes to the goal of optimal tracking control through minimization of an associated infinite-horizon performance cost
\begin{equation}
\label{eq:3}
J^{\mu_i,\mu_{-i}}_{i}(x_0,r_0) = \sum_{k=0}^{\infty} \bigg[\gamma^{k}l_{i}\big(x_k,r_k,\mu_{i}(x_k,r_k),\mu_{-i}(x_k,r_k)\big)\bigg],
\end{equation}
where $\gamma \in (0,1]$ is the discount factor and $l_{i}:\mathcal{Z}_i\rightarrow \mathbb{R}_{+}$ is the stage cost function
\begin{align}
\label{eq:4}
&l_{i}\big(x,r,\mu_{i}(x,r),\mu_{-i}(x,r)\big) = (x-r)^{T}S_{ii}(x-r)\nonumber \\&+ \sum_{j=1}^{N} \mu_{j}^{T}(x,r)R_{ij}\mu_{j}(x,r)
\end{align}
with $S_{ii} \in \mathbb{S}_{++}^{n}$ and $R_{ij} \in \mathbb{S}_{++}^{m_i}$. We note that a discount factor $\gamma \in (0,1)$ is required to ensure that \eqref{eq:3} is finite and $\gamma=1$ can only be used if the reference dynamics \eqref{eq:2} are asympotically stable \cite{b72,b73,b74}.\\
We are interested in the case where the mathematical expressions of the system \eqref{eq:1} and reference input \eqref{eq:2} are unknown, although their values can be observed through simulations and experiments. Towards this direction, we define the Q-function $Q^{\mu_{i},\mu_{-i}}_{i}:\mathcal{Z}_i\rightarrow \mathbb{R}_{+}$ for all $i$ as
\begin{align}
\label{eq:5}
Q^{\mu_{i},\mu_{-i}}_{i}(x_k,r_k,a_{ik},a_{-ik}) =& l_{i}(x_k,r_k,a_{ik},a_{-ik})\nonumber \\
&+\gamma J^{\mu_{i},\mu_{-i}}_{i}(x_{k+1},r_{k+1}).
\end{align}
The Q-function encodes the cost of applying control inputs $(a_{i},a_{-i})\in \mathcal{U}_{i}\times \mathcal{U}_{-i}$ at state $x$ and for reference $r$ and then following control policies $\{\mu_{i},\mu_{-i}\}$ afterwards.
\begin{assumption}
\label{ass:2}
$Q^{\mu_{i},\mu_{-i}}_{i}(x,r,a_i,a_{-i})$ is a continuously differentiable function on $\mathcal{Z}_i$ for all $i \in \mathcal{P}$.
\end{assumption}
Clearly,
\begin{equation}
\label{eq:6}
Q^{\mu_{i},\mu_{-i}}_{i}\big(x,r,\mu_{i}(x,r),\mu_{-i}(x,r)\big) = J^{\mu_{i},\mu_{-i}}_{i}(x,r).
\end{equation}
Therefore, \eqref{eq:5} can be rewritten as 
\begin{align}
\label{eq:7}
&Q^{\mu_{i},\mu_{-i}}_{i}(x_k,r_k,a_{ik},a_{-ik}) = l_{i}(x_k,r_k,a_{ik},a_{-ik}) \nonumber\\ 
&+\gamma Q^{\mu_{i},\mu_{-i}}_{i}\big(x_{k+1},r_{k+1},\mu_{i}(x_{k+1},r_{k+1}),\mu_{-i}(x_{k+1},r_{k+1})\big).
\end{align}
For all $i\in \mathcal{P}$, the optimal Q-function $Q^{\star}_{i}$ satisfies the discrete-time coupled Bellman optimality equation \cite{b1,b11,b15,b20,b22} 
\begin{align}
\label{eq:8}
&Q^{\star}_{i}(x_k,r_k,a_{ik},a_{-ik}) = l_{i}(x_k,r_k,a_{ik},a_{-ik})\nonumber \\ 
&+\gamma \min_{u_i}Q^{\star}_{i}\big(x_{k+1},r_{k+1},u_{i},\mu_{-i}(x_{k+1},r_{k+1})\big)
\end{align}
and the associated optimal control policy is given by
\begin{equation}
\label{eq:9}
\mu^{\star}_{i}(x,r) = \underset{u_{i}}{\mathrm{argmin}}\enspace Q^{\star}_{i}\big(x,r,u_{i},\mu_{-i}(x,r)\big).  
\end{equation}
In NZSGs, all players $i\in \mathcal{P}$ have the same competitive hierarchical level and try to achieve optimal control through convergence to a Nash equilibrium (NE), defined as follows.
\begin{defi}[\cite{b1,b4}]
\label{def:1}
A set of control policies $\{\mu^{\star}_{i}(x,r)\}_{i=1}^{N}$ is said to constitute a Nash equilibrium (NE) for the discrete-time nonlinear system \eqref{eq:1}, if for all $i\in \mathcal{P}$:
\begin{equation}
\label{eq:10}
J^{\star}_{i} = J^{\mu^{\star}_i,\mu^{\star}_{-i}}_{i}(x,r)\leq  J^{u_i,\mu^{\star}_{-i}}_{i}(x,r), \enspace \text{for all } u_i \in \mathcal{U}_i.
\end{equation}
\end{defi}
The following theorem provides some important results on closed-loop stability and the derivation of NE solutions for the considered class of discrete-time multiplayer NZSGs \eqref{eq:1}.
\begin{theorem}
\label{thm:1}
Let Assumptions 1 and 2 hold. Assume that there exists a positive definite solution $Q^{\star}_i$ of the discrete-time coupled Bellman optimality equation
\begin{align}
\label{eq:11}
&Q^{\star}_{i}(x_k,r_k,a_{ik},a_{-ik}) = l_{i}(x_k,r_k,a_{ik},a_{-ik})\nonumber \\ 
&+\gamma Q^{\star}_{i}\big(x_{k+1},r_{k+1},\mu^{\star}_{i}(x_{k+1},r_{k+1}),\mu^{\star}_{-i}(x_{k+1},r_{k+1})\big)
\end{align} 
for all $i\in \mathcal{P}$, where
\begin{equation}
\label{eq:12}
\mu^{\star}_{i}(x,r) = \underset{u_i}{\mathrm{argmin}}\enspace Q^{\star}_{i}\big(x,r,u_{i},\mu^{\star}_{-i}(x,r)\big)
\end{equation}
and $\mu^{\star}_{i}(0,0)=0$. Let also define the tracking error as $e_k=x_k-r_k$. Then, for all time steps $s \geq k+1$:
\begin{enumerate}
\item[1)] If $\gamma=1$, the set of control policies $\{\mu^{\star}_{i}(x,r)\}_{i=1}^{N}$ can make the tracking error $e_s$ for the nonlinear system \eqref{eq:1} locally asymptotically stable. Otherwise, $\{\mu^{\star}_{i}(x,r)\}_{i=1}^{N}$ can make $e_s$ for \eqref{eq:1} sufficiently small by setting $\gamma$ sufficiently close to 1.
\item[2)] The set of control policies $\{\mu^{\star}_{i}(x,r)\}_{i=1}^{N}$ constitutes a NE solution for \eqref{eq:1}.
\item[3)] The NE outcome for each player $i\in \mathcal{P}$ is given by
\begin{align*}
J^{\star}_{i} = Q^{\star}_{i}(x_s,r_s,\mu^{\star}_{i}(x_s,r_s),\mu^{\star}_{-i}(x_s,r_s)).
\end{align*}
\end{enumerate}
\end{theorem}
\begin{proof}
To simplify the presentation of the proof, we introduce the following compact notation
\begin{align}
\label{eq:13}
F_i^{\mu_{i},\mu_{-i}}(x,r) =& F_i\big(x,r,\mu_{i}(x,r),\mu_{-i}(x,r)\big)\nonumber \\
F_i^{\star,\mu_{i},\mu_{-i}}(x,r) =& F_i^{\star}\big(x,r,\mu_{i}(x,r),\mu_{-i}(x,r)\big)
\end{align}
for generic functions $F_i:\mathcal{Z}_i\rightarrow \mathbb{R}_{+}$ and $F_i^{\star}:\mathcal{Z}_i\rightarrow \mathbb{R}_{+}$ for all $i\in \mathcal{P}$.
\begin{itemize}
\item[1)] Let $\gamma^{s}Q^{\star,\mu^{\star}_{i},\mu^{\star}_{-i}}_{i}(x_s,r_s)$ be a candidate Lyapunov function for all $i \in \mathcal{P}$. By defining the difference equation as $D_i\big(\gamma^{s}Q^{\star,\mu^{\star}_{i},\mu^{\star}_{-i}}_{i}(x_s,r_s)\big)$$ = \gamma^{s+1}Q^{\star,\mu^{\star}_{i},\mu^{\star}_{-i}}_{i}(x_{s+1},r_{s+1})-\gamma^{s}Q^{\star,\mu^{\star}_{i},\mu^{\star}_{-i}}_{i}(x_s,r_s)$, we get
\begin{align}
\label{eq:14} 
D_i\big(\gamma^{s}Q^{\star,\mu^{\star}_{i},\mu^{\star}_{-i}}_{i}(x_s,r_s)\big)
&=-\gamma^{s}l_i^{\mu^{\star}_{i},\mu^{\star}_{-i}}(x_s,r_s) \leq 0.
\end{align}
If $\gamma=1$ (which can only be used if $\lim_{s\to\infty} r_{s}\rightarrow 0$), then according to Barbalat's Extension Lemma \cite{b74},\cite[p. 113]{b75}, the states $x$ of \eqref{eq:1} converge in a region where $\lim_{s\to \infty} D_i\big(\gamma^{s}Q^{\star,\mu^{\star}_{i},\mu^{\star}_{-i}}_{i}(x_s,r_s)\big) \rightarrow 0$. Based on \eqref{eq:14}, local asymptotic stability of the tracking error $e_s$ on $\mathcal{X}$ for the closed-loop system \eqref{eq:1} is proved under $\{\mu^{\star}_{i}(x,r)\}_{i=1}^{N}$, i.e., $\lim_{s\to\infty} x_{s}\rightarrow 0$. Otherwise, following \cite{b74,b76}, the tracking error $e_s$ for \eqref{eq:1} can be made sufficiently small by setting $\gamma$ sufficiently close to $1$.
\item[2)] For all $i\in \mathcal{P}$, we define
\begin{align*}
J^{\mu_i,\mu^{\star}_{-i}}_{i}(x_s,r_s) =& \sum_{m=s}^{\infty} \gamma^{(m-s)}l^{\mu_{i},\mu^{\star}_{-i}}_{i}(x_m,r_m).
\end{align*}
By adding and substracting $Q^{\star,\mu^{\star}_i,\mu^{\star}_{-i}}_{i}(x_s,r_s)$ in the right hand side yields
\begin{align*}
J^{\mu_i,\mu^{\star}_{-i}}_{i}(x_s,r_s) =& \sum_{m=s}^{\infty} \gamma^{(m-s)}l^{\mu_{i},\mu^{\star}_{-i}}_{i}(x_m,r_m)\\
-&Q^{\star,\mu^{\star}_i,\mu^{\star}_{-i}}_{i}(x_s,r_s)\\
+& Q^{\star,\mu^{\star}_i,\mu^{\star}_{-i}}_{i}(x_s,r_s)\\
\overset{\eqref{eq:12}}{\geq} & Q^{\star,\mu^{\star}_i,\mu^{\star}_{-i}}_{i}\big(x_s,r_s)\\
=& J^{\mu^{\star}_i,\mu^{\star}_{-i}}_{i}(x_s,r_s).
\end{align*}
Based on Definition 1, the set of control policies $\{\mu^{\star}_{i}(x,r)\}_{i=1}^{N}$ constitutes a NE solution for \eqref{eq:1}.
\item[3)] According to \eqref{eq:6} and \eqref{eq:10}, it is easy to observe that
\begin{align*}
J^{\star}_{i} =& J^{\mu^{\star}_i,\mu^{\star}_{-i}}_{i}(x_{s},r_{s})\\
=& Q^{\star,\mu^{\star}_{i},\mu^{\star}_{-i}}_{i}(x_s,r_s)
\end{align*}
for all $i\in \mathcal{P}$.
\end{itemize} 
\end{proof}
\section{Multi-Step VI For Model-Free Optimal Tracking Control Of NZSGs}
In this section, we derive a novel Q-function-based multi-step VI algorithm for model-free optimal tracking control of discrete-time multiplayer NZSGs. Algorithm 1 shows the proposed method, which we call MSQVI.
\begin{algorithm}
        \caption{The proposed MSQVI algorithm.}\label{alg1}
        \begin{algorithmic}[1]
            \State  \textbf{Initialization:} Choose $Q^{0}_{i}(x,r,a_{i},a_{-i})\geq 0$ for all $i$ and arbitrary $\{\mu^{-1}_{i}(x,r)\}_{i=1}^{N}$. Set $p=0$ and $\tau\geq 0$.
            \State\textbf{Policy Improvement:} 
\begin{equation}
\label{eq:15}
\mu^{p}_{i}(x,r) = \underset{u_{i}}{\mathrm{argmin}}\enspace Q^{p}_{i}\big(x,r,u_{i},\mu^{p-1}_{-i}(x,r)\big) \enspace \forall i.
\end{equation}
            \State \textbf{Policy Evaluation:} For all $i$, solve for $Q^{p+1}_{i}\geq 0$,
\begin{align}
\label{eq:16}
&Q^{p+1}_{i}(x,r,a_{i},a_{-i})= l_{i}(x,r,a_{i},a_{-i})\nonumber \\
&+\sum_{m=1}^{H_p-1}\gamma^{m}l_{i}\big(x_m,r_m,\mu^{p}_{i}(x_m,r_m),\mu^{p}_{-i}(x_m,r_m)\big)\nonumber \\
&+\gamma^{H_p} Q^{p}_{i}\big(x_{H_p},r_{H_p},\mu^{p}_{i}(x_{H_p},r_{H_p}),\mu^{p}_{-i}(x_{H_p},r_{H_p})\big),
\end{align}
where $x_{1} = f(x)+\sum_{j=1}^{N} g_{j}(x)a_{j}$, $r_1= h(r)$, $x_{m} = f(x_{m-1})+\sum_{j=1}^{N} g_{j}(x_{m-1})\mu^{p}_{j}(x_{m-1},r_{m-1})$ and $r_m = h(r_{m-1})$ for $m>1$.
             \State \textbf{Termination of Learning Phase:} If $\|Q^{p+1}_{i}-Q^{p}_{i}\|_{\infty}\leq \tau$ for all $i$, then terminate; else set $p=p+1$, go to Step $2$ and continue.
        \end{algorithmic}
    \end{algorithm}
In contrast to conventional PI \cite{b9,b10,b11,b12,b13,b14,b15,b16,b17,b18,b19,b20,b21,b22} and VI \cite{b23,b24,b25,b26,b27} algorithms, the policy evaluation of MSQVI \eqref{eq:16} utilizes finite lookahead data, defined by $H_p\geq 1$; by setting $H_p=1$ for all $p$, MSQVI is converted to the VI algorithm. The initialization of MSQVI requires the choice of suitable Q-functions $Q^{0}_{i}\geq 0$ for all $i$ (to be discussed in the sequel) and an arbitrary set of control policies $\{\mu^{-1}_{i}(x,r)\}_{i=1}^{N}$. The algorithm utilizes a game theoretical setup widely used in the literature \cite{b9,b10,b11,b12,b13,b14,b15,b16,b17,b18,b19,b20,b21,b22,b23,b24,b25,b26,b27}. The policy evaluation scheme \eqref{eq:16} initiates round $p$ of the NZSG. All players $i\in \mathcal{P}$ interact with the nonlinear system \eqref{eq:1} by first applying exploratory actions $\{a_i\}_{i=1}^{N}$ and then employing their most recently updated control policies $\{\mu^{p}_i(x,r)\}_{i=1}^{N}$ for a finite time horizon defined by $H_p$. During a game round, all players $i\in \mathcal{P}$ have access to the state and reference values $(x,r)$. The buffer of historical data values $\big\{x,r,a_i,a_{-i},[x_m,r_m,\mu^{p}_i(x_m,r_m),\mu^{p}_{-i}(x_m,r_m)]_{m=1}^{H_p}\big\}$ is then broadcasted to all players for the computation of $Q^{p+1}_{i}\geq 0$ for all $i$ as in \eqref{eq:16}. This concludes the game round $p$. The explicit functional forms of all control policies $\{\mu_{i}^{p}\}_{i=1}^{N}$ previously employed during the game round are then broadcasted to all players so that they can proceed to policy improvement \eqref{eq:15}. We now proceed with the theoretical analysis of the proposed MSQVI algorithm.
\begin{theorem}
Let Assumptions 1 and 2 hold. For $i\in \mathcal{P}$, consider the sequences $\{Q^{p}_{i}(x,r,a_{i},a_{-i})\}_{p \in \mathbb{N}}$ and $\{\mu^{p}_{i}(x,r)\}_{p \in \mathbb{N}_0}$ generated by Algorithm 1. Assume that the initialization condition
\begin{align}
\label{eq:17}
Q^{0}_{i}(x,r,a_{i},a_{-i})\geq & l_{i}(x,r,a_{i},a_{-i})\nonumber \\
&+\gamma Q^{0}_{i}\big(x_1,r_1,\mu^{0}_{i}(x_1,r_1),\mu^{0}_{-i}(x_1,r_1)\big)
\end{align}
holds for all $i\in \mathcal{P}$ and $(x,r,a_i,a_{-i})\in \mathcal{Z}_i$, $\|g_j(x)\|_{2}$ and $\overline{\sigma}(R_{ij}R^{-1}_{jj})$ are sufficiently small for all $j\in \mathcal{P}\setminus \{i\}$. Then, for all $i\in \mathcal{P}$ and $(x,r,a_i,a_{-i})\in \mathcal{Z}_i$
\begin{itemize}
\item[1)]
\begin{align}
\label{eq:18}
Q^{p+1}_{i}(x,r,a_{i},a_{-i})\leq  & l_{i}(x,r,a_{i},a_{-i})\nonumber \\
+& \gamma Q^{p}_{i}\big(x_1,r_1,\mu^{p}_{i}(x_1,r_1),\mu^{p}_{-i}(x_1,r_1)\big)\nonumber \\
\leq & Q^{p}_{i}(x,r,a_{i},a_{-i}).
\end{align}
\item[2)]
\begin{align*}
\lim_{p\rightarrow \infty}Q^{p}_{i}(x,r,a_{i},a_{-i})=& Q^{\star}_{i}(x,r,a_{i},a_{-i}), \\
\lim_{p\rightarrow \infty}\mu^{p}_{i}(x,r)=& \mu^{\star}_{i}(x,r).
\end{align*}
\end{itemize}
\end{theorem}

\begin{proof}
See Appendix I.
\end{proof}
\begin{crl}
For all $i\in \mathcal{P}$, consider $\tilde{Q}_{i}(x,r,a_{i},a_{-i})$ and the set of control policies $\{\tilde{\mu}_{i}(x,r)\}_{i=1}^{N}$ which are computed based on \eqref{eq:15}. Let $h_1\geq h_2\geq 1$ be two horizon lengths, leading to policy evaluation \eqref{eq:16} given by
\begin{equation*}
\begin{aligned}
&Q_{i,H}(x,r,a_{i},a_{-i}) = l_{i}(x,r,a_{i},a_{-i})\nonumber\\
&+\sum_{m=1}^{H-1}\gamma^{m}l_{i}\big(x_m,r_m,\tilde{\mu}_{i}(x_m,r_m),\tilde{\mu}_{-i}(x_m,r_m)\big) \nonumber \\
&+\gamma^{H} \tilde{Q}_{i}\big(x_{H},r_{H},\tilde{\mu}_{i}(x_{H},r_{H}),\tilde{\mu}_{-i}(x_{H},r_{H})\big),
\end{aligned}
\end{equation*}
with $H\in \{h_1,h_2\}$. Then it holds that $Q_{i,h_1}(x,r,a_{i},a_{-i})\leq Q_{i,h_2}(x,r,a_{i},a_{-i})$ for all $(x,r,a_i,a_{-i})\in \mathcal{Z}_i$.
\end{crl}
\begin{proof}
See Appendix II.
\end{proof}
\textit{Remark 1:} Similarly to multi-step VI approaches for single-input systems \cite{b39,b40}, to ensure \eqref{eq:17}, it suffices to initialize $Q^{0}_{i}(x,r,a_{i},a_{-i})$ with a sufficiently large, continuously differentiable, positive definite function for all $i$. Furthermore, the conditions on $\|g_j(x)\|_2$ and $\overline{\sigma}(R_{ij}R^{-1}_{jj})$ for $j\in \mathcal{P}\setminus \{i\}$ theoretically ensure the monotonic convergence of MSQVI to NE solutions for the class of weakly coupled games \cite{b4}. Similar conditions are required by PI and VI to ensure monotonicity and convergence \cite{b9,b13,b15,b18,b21,b22,b23,b77}. In contrast to PI and VI, according to Corollary 1, the proposed MSQVI algorithm can potentially enable improved quality of solutions and/or faster convergence also for more general NZSGs by increasing the horizon length $H_p$. Finally we note that, by setting $H_p=1$ for all $p$ in Theorem 2, we can similarly prove the non-increasing monotonic convergence to NE solutions for the case of VI, previously explicitly proved only for discrete-time linear NZSGs and graphical games \cite{b23,b24,b26,b27}.
\section{A Data-Driven Implementation}
In this section, a data-driven implementation approach based on critic-only LS is developed for the proposed MSQVI algorithm. We consider a linearly parameterized function approximation for $Q^{p}_{i}(x,r,a_{i},a_{-i})$ on $\mathcal{Z}_i$ as
\begin{equation}
\label{eq:19}
Q^{p}_{i}(x,r,a_{i},a_{-i}) = [\Phi_i(x,r,a_{i},a_{-i})]^{T} w^{p}_{i}+e^{p}_{i}(x,r,a_{i},a_{-i}) 
\end{equation}  
for all $i$, where $w^{p}_{i} \in \mathbb{R}^{K}$ is the vector of weights, $\Phi_i(x,r,a_{i},a_{-i})\in \mathbb{R}^{K}$ the vector of linearly independent polynomial basis functions and $e^{p}_{i}(x,r,a_{i},a_{-i})\in \mathbb{R}$ the associated error of approximation. Since $\mathcal{Z}_i$ is compact, the Stone-Weierstrass Theorem \cite{b78,b79,b80} ensures that $w_i$ can be selected so that $\lim_{K\to \infty} e^{p}_{i}(x,r,a_{i},a_{-i}) = 0$.  For implementation purposes, due to the fact that $w^{p}_{i}$ is unknown, we define 
\begin{equation*}
\hat{Q}^{p}_{i}(x,r,a_{i},a_{-i}) = [\Phi_i(x,r,a_{i},a_{-i})]^{T}\hat{w}^{p}_{i}
\end{equation*}
for all $i$, where $\hat{w}^{p}_{i}\in \mathbb{R}^{K}$ is an estimation of $w^{p}_{i}$. Then, the policy improvement \eqref{eq:15} becomes
\begin{equation*}
\hat{\mu}^{p}_{i}(x,r) = \underset{u_{i}}{\mathrm{argmin}}\enspace \hat{Q}^{p}_{i}\big(x,r,u_{i},\hat{\mu}^{p-1}_{-i}(x,r)\big).
\end{equation*}
Furthermore, the policy evaluation scheme \eqref{eq:16} is now given by
\begin{equation}
\begin{aligned}
\label{eq:20}
&\epsilon^{p+1}_{i}(x,r,a_{i},a_{-i}) = \hat{Q}^{p+1}_{i}(x,r,a_{i},a_{-i})-l_{i}(x,r,a_{i},a_{-i}) \\
&- \sum_{m=1}^{H_{p}-1}\gamma^{m}l_{i}\big(x_m,r_m,\hat{\mu}^{p}_{i}(x_m,r_m),\hat{\mu}^{p}_{-i}(x_m,r_m)\big) \\
&-\gamma^{H_{p}}\hat{Q}^{p}_{i}\big(x_{H_p},r_{H_p},\hat{\mu}^{p}_{i}(x_{H_p},r_{H_p}),\hat{\mu}^{p}_{-i}(x_{H_p},r_{H_p})\big) \\
&= [\Phi_{i}(x,r,a_{i},a_{-i})]^{T}\hat{w}^{p+1}_{i} -l_{i}(x,r,a_{i},a_{-i}) \\
&- \sum_{m=1}^{H_{p}-1}\gamma^{m}l_{i}\big(x_{m},r_m,\hat{\mu}^{p}_i(x_m,r_m),\hat{\mu}^{p}_{-i}(x_m,r_m)\big)\\
&- \gamma^{H_p}[\Phi_{i}\big(x_{H_p},r_{H_p},\hat{\mu}^{p}_{i}(x_{H_p},r_{H_p}),\hat{\mu}^{p}_{-i}(x_{H_p},r_{H_p})\big)]^{T}\hat{w}^{p}_{i},
\end{aligned}
\end{equation}
where $\epsilon^{p+1}_{i}(x,r,a_{i},a_{-i})$ is the residual error due to the approximation errors $e^{p+1}_i$ on $\hat{Q}^{p+1}_{i}$ and $e^{p}_i$ on $\hat{Q}^{p}_{i}$ . Based on \eqref{eq:20}, the unknown vector $\hat{w}^{p+1}_{i}$ is computed by collecting system data. For all players $i\in \mathcal{P}$ and $p\geq 0$, let
\begin{equation}
\begin{aligned}
\label{eq:21}
S^{p}_{i} =& \big\{x_b,r_b,a_{ib},a_{-ib},[x_{m,b},r_{m,b},\hat{\mu}^{p}_i(x_{m,b},r_{m,b}), \\
&\hat{\mu}^{p}_{-i}(x_{m,b},r_{m,b})]_{m=1}^{H_p}\big\}_{b=1}^{B}
\end{aligned}
\end{equation}
be a buffer constructed from the data of the game round $p$, where $x_{1,b} = f(x_{b})+\sum_{j=1}^{N} g_{j}(x_{b})a_{jb}$, $r_{1,b} = h(r_b)$,  $x_{m,b} = f(x_{m-1,b})+\sum_{j=1}^{N} g_{j}(x_{m-1,b})\hat{\mu}^{p}_{j}(x_{m-1,b},r_{m-1,b})$, $r_{m,b} = h(r_{m-1,b})$ for $m>1$, while $B\in \mathbb{N}$ is the size of the buffer. The residual error is then given by
\begin{equation}
\begin{aligned}
\label{eq:22}
&\epsilon^{p+1}_{i,b}(x_b,r_b,a_{ib},a_{-ib}) = [\Phi_{i}(x_b,r_b,a_{ib},a_{-ib})]^{T}\hat{w}^{p+1}_{i} - \gamma^{H_p}\cdot \\
&\big[\Phi_{i}\big(x_{H_p,b},r_{H_p,b},\hat{\mu}^{p}_{i}(x_{H_p,b},r_{H_p,b}),\hat{\mu}^{p}_{-i}(x_{H_p,b},r_{H_p,b})\big)\big]^{T}\hat{w}^{p}_{i} \\
&- \sum_{m=1}^{H_{p}-1}\gamma^{m}l_{i}\big(x_{m,b},r_{m,b},\hat{\mu}^{p}_{i}(x_{m,b},r_{m,b}),\hat{\mu}^{p}_{-i}(x_{m,b},r_{m,b})\big) \\
&-l_{i}(x_b,r_b,a_{ib},a_{-ib})
\end{aligned}
\end{equation}
for $b=1,\ldots,B$. The unknown vector $\hat{w}^{p+1}_{i}$ can then be computed by minimizing the sum of residual errors, that is
\begin{equation}
\label{eq:23}
\min \sum_{b=1}^{B} (\epsilon^{p+1}_{i,b})^{2}.
\end{equation}
Then, the least squares scheme is implemented as follows
\begin{equation}
\label{eq:24}
\hat{w}^{p+1}_{i} = [\Psi^{T}_{i}\Psi_{i}]^{-1}\Psi^{T}_{i} z^{p}_{i}
\end{equation}
for all $i$, where $z^{p}_{i} = \begin{bmatrix} z^{p}_{i,1} & \cdots & z^{p}_{i,B}\end{bmatrix}^{T}$, $\Psi_{i} = \begin{bmatrix} \Psi_{i,1} & \cdots & \Psi_{i,B}\end{bmatrix}^{T}$, $\Psi_{i,b} = \Phi_{i}(x_{b},r_b,a_{ib},a_{-ib})$, $z^{p}_{i,b} =\gamma^{H_p}[\Phi_{i}\big(x_{H_p,b},r_{H_p,b},\hat{\mu}^{p}_{i}(x_{H_p,b},r_{H_p,b}),\hat{\mu}^{p}_{-i}(x_{H_p,b},r_{H_p,b})\big)]^{T}\cdot$\\$\hat{w}^{p}_{i}+ \sum_{m=1}^{H_{p}-1}\gamma^{m}l_{i}\big(x_{m,b},r_{m,b},\hat{\mu}^{p}_{i}(x_{m,b},r_{m,b}),\hat{\mu}^{p}_{-i}(x_{m,b},r_{m,b})\big)$\\$+l_{i}(x_b,r_b,a_{ib},a_{-ib})$. Algorithm 2 shows the data-driven LS implementation of the MSQVI algorithm (Algorithm 1), which we refer to as MSQVI-LS algorithm.
\begin{algorithm}
        \caption{The proposed MSQVI-LS algorithm.}\label{alg2}
        \begin{algorithmic}[1]
		 \State \textbf{Initialization:} Define $\{\hat{Q}_{i}^{0}\}_{i=1}^{N}=\{Q_{i}^{0}\}_{i=1}^{N}$ based on Remark 1 and arbitrary $\{\hat{\mu}_{i}^{-1}(x,r)\}_{i=1}^{N}=\{\mu_{i}^{-1}(x,r)\}_{i=1}^{N}$. Set $p=0$, $\tau\geq 0$ and $B\in \mathbb{N}$.
               \State \textbf{Policy Improvement:}\newline $\hat{\mu}^{p}_{i}(x,r)=\underset{u_{i}}{\mathrm{argmin}}\enspace \hat{Q}^{p}_{i}\big(x,r,u_{i},\hat{\mu}^{p-1}_{i}(x,r)\big)$, for all $i$.
             \State \textbf{Choice of horizon length:} Select $H_p\geq 1$.
             \State \textbf{Data collection:} Construct data buffer $S^{p}_{i}$ \eqref{eq:21}, for all $i$.
             \State \textbf{Policy Evaluation:} Solve \eqref{eq:24} for $\hat{w}^{p+1}_{i}$, for all $i$.
\State \textbf{Termination of Learning Phase:}\newline
If $\underset{b}{\mathrm{max}} |\hat{Q}^{p+1}_{i}(x_b,r_b,a_{ib},a_{-ib})-\hat{Q}^{p}_{i}(x_b,r_b,a_{ib},a_{-ib})|> \tau$ for $i\in \mathcal{P}$, set $p=p+1$ and go to Step $2$. Otherwise, set $\{\hat{Q}^{\star}_{i}(x,r,a_i,a_{-i})\}_{i=1}^{N}=\{\hat{Q}^{p+1}_{i}(x,r,a_i,a_{-i})\}_{i=1}^{N}$ and return $\{\hat{\mu}^{\star}_i(x,r)\}_{i=1}^{N}= \{\hat{\mu}^{p+1}_i(x,r)\}_{i=1}^{N}$ as the set of approximate optimal control policies.
        \end{algorithmic}
    \end{algorithm}
\begin{theorem}
Let Assumptions 1 and 2 hold. For $i\in \mathcal{P}$, consider the sequences $\{\hat{Q}^{p}_{i}(x,r,a_{i},a_{-i})\}_{p\in \mathbb{N}}$ and $\{\hat{\mu}^{p}_{i}(x,r)\}_{p\in \mathbb{N}_0}$ generated by Algorithm $2$. Assume that there exist constants $\bar{B}>0$ and $\delta >0$ such that for all $B\geq \bar{B}$
\begin{equation}
\label{eq:25}
\frac{1}{B} \sum_{b=1}^{B} \Psi_{i,b}\Psi^{T}_{i,b} \geq \delta I_{B}.
\end{equation}
Then, $\lim_{p,K\to \infty} \hat{Q}^{p}_{i}(x,r,a_{i},a_{-i}) = Q^{\star}_{i}(x,r,a_{i},a_{-i})$ and $\lim_{p,K\to \infty} \hat{\mu}^{p}_{i}(x,r) = \mu^{\star}_{i}(x,r) $, which satisfy \eqref{eq:11} for all $i$.
\end{theorem}
\begin{proof}
See Appendix III.
\end{proof}
\textit{Remark 2:} Based on Theorem 3, to guarantee convergence of the set of weight vectors $\{\hat{w}^{p}_i\}_{i=1}^{N}$, the persistence of excitaton (PoE) condition \eqref{eq:25} \cite{b81} is required to hold for all $i\in \mathcal{P}$. This condition ensures the existence of the inverse of the matrix $\Psi^{T}_{i}\Psi_{i}$ in \eqref{eq:24} for all $i$. For practical implementation,  to ensure \eqref{eq:25}, we can apply randomized policies $\{a_{ik}\}_{i=1}^{N}$ to the system \eqref{eq:1} (e.g., randomized experience replay \cite{b40,b45}) or employ general off-policy learning methods \cite{b82,b83}, where $\{a_{ik}\}_{i=1}^{N}$ acts as a set of appropriate exploration policies applied to the system \eqref{eq:1} and differ from the set of evaluated control policies $\{\hat{\mu}^{p}_i(x,r)\}_{i=1}^{N}$. The richness of $\{a_{ik}\}_{i=1}^{N}$ and particular choice of buffer size $B$ is generally dependent on the complexity of the considered system \eqref{eq:1}. 
\section{A Linear Programming Approach For NZSGs}
In this section, we proceed with the transformation of the policy evaluation scheme \eqref{eq:16} into a tractable data-driven optimization problem. We define $\mathcal{F}(\mathcal{Z}_i)$ as a vector space of bounded (in a suitably weighted norm), real-valued, Borel-measurable functions on $\mathcal{Z}_i$ \cite{b43,b48}. We now proceed by introducing the following functional operator.
\begin{defi}
For $i\in \mathcal{P}$ and a given set of control policies $\{\mu_{-i}(x,r)\}$, the coupled Bellman operator is the mapping $\mathcal{C}_i:\mathcal{F}(\mathcal{Z}_i)\rightarrow \mathcal{F}(\mathcal{Z}_i)$ defined as
\begin{equation}
\begin{aligned}
\label{eq:26}
&\mathcal{C}_iQ_{i}(x_k,r_k,a_{ik},a_{-ik}) = l_{i}(x_k,r_k,a_{ik},a_{-ik}) \\
&+\gamma \min_{u_{i}} Q_{i}\big(x_{k+1},r_{k+1},u_{i},\mu_{-i}(x_{k+1},r_{k+1})\big).
\end{aligned}
\end{equation}
\end{defi}
Please note that dependence of $\mathcal{C}_i$ on $\{\mu_{-i}(x,r)\}$ is suppressed to simplify notation. The operator $\mathcal{C}_i$ not only retains the same structure as the standard Bellman operator, but also inherits its monotone contraction properties, as shown in the following proposition. 
\begin{proposition}
For all $i\in \mathcal{P}$ and $\{\mu_{-i}(x,r)\}$, the coupled Bellman operator $\mathcal{C}_i$ is a monotone contraction mapping with a unique fixed point in $\mathcal{F}(\mathcal{Z}_i)$.
\end{proposition}
\begin{proof}
We use the compact notation \eqref{eq:13} to present the proof. For all $i$, we first consider $Q_{i1}, Q_{i2} \in \mathcal{F}(\mathcal{Z}_i)$ such that
\begin{equation*}
Q^{\mu_{i},\mu_{-i}}_{i1}(x,r) \leq Q^{\mu_{i},\mu_{-i}}_{i2}(x,r) \enspace \forall (i,x,r,\mu_{i},\mu_{-i}).
\end{equation*}
Then, we get
\begin{equation*}
\begin{aligned}
Q^{\mu_{i},\mu_{-i}}_{i1}(x_{k+1},r_{k+1}) \leq & Q^{\mu_{i},\mu_{-i}}_{i2}(x_{k+1},r_{k+1})\\
\min_{u_{i}} Q^{u_{i},\mu_{-i}}_{i1}(x_{k+1},r_{k+1}) \leq & \min_{u_{i}} Q^{u_{i},\mu_{-i}}_{i2}(x_{k+1},r_{k+1}) \\
\mathcal{C}_iQ_{i1}(x_k,r_k,a_{ik},a_{-ik}) \leq & \mathcal{C}_iQ_{i2}(x_k,r_k,a_{ik},a_{-ik})
\end{aligned}
\end{equation*}
for all $(x_k,r_k,a_{ik},a_{-ik})\in \mathcal{Z}_i$. Therefore, the operator $\mathcal{C}_i$ is monotone \cite{b84}.\\
Next, for all $i\in \mathcal{P}$, given $Q_{i1},Q_{i2} \in \mathcal{F}(\mathcal{Z}_i)$, we have that
\begin{equation*}
\begin{aligned}
&\big|\mathcal{C}_iQ_{i1}(x_k,r_k,a_{ik},a_{-ik})-\mathcal{C}_iQ_{i}(x_k,r_k,a_{ik},a_{-ik})\big|= \\
&\gamma \big|\min_{u_{i}} Q^{u_{i},\mu_{-i}}_{i1}(x_{k+1},r_{k+1})-\min_{u_{i}} Q^{u_{i},\mu_{-i}}_{i2}(x_{k+1},r_{k+1})\big|\\
& \leq  \gamma \max_{u_{i}} \big|Q^{u_{i},\mu_{-i}}_{i1}(x_{k+1},r_{k+1}) - Q^{u_{i},\mu_{-i}}_{i2}(x_{k+1},r_{k+1}) \big| \\
& \leq \gamma \max_{x_k,r_k,u_{i},u_{-i}} \big| Q^{u_{i},u_{-i}}_{i1}(x_k,r_k) - Q^{u_{i},u_{-i}}_{i2}(x_k,r_k) \big|\\
&= \gamma \max_{x_k,r_k,a_{ik},a_{-ik}} \big| Q_{i1}(x_k,r_k, a_{ik},a_{-ik})- Q_{i2}(x_k,r_k, a_{ik},a_{-ik}) \big|,
\end{aligned}
\end{equation*}
i.e., $\|\mathcal{C}_iQ_{i1}-\mathcal{C}_iQ_{i2}\|_{\infty} \leq \gamma \|Q_{i1}-Q_{i2}\|_{\infty}$ for all $ Q_{i1},Q_{i2} \in \mathcal{F}(\mathcal{Z}_i)$ and $i\in \mathcal{P}$. Hence, $\mathcal{C}_i$ is a $\gamma$-contraction with respect to the max norm \cite[Def. 5.1-1]{b85}. As $\mathcal{F}(\mathcal{Z}_i)$ is complete under the weighted sup norm, the uniqueness of fixed point in \eqref{eq:26} follows \cite[Thm 5.1-2]{b85}.
\end{proof}
Based on Proposition 1, if $Q_{i}\in \mathcal{F}(\mathcal{Z}_i)$ satisfies the coupled Bellman inequality $Q_i \leq \mathcal{C}_iQ_i$ for all $i$, the monotone contraction property of $\mathcal{C}_i$ implies
\begin{equation*}
Q_i \leq \mathcal{C}_iQ_i \leq \ldots \leq \lim_{p\to \infty} \mathcal{C}^{p}_iQ_i =Q^{\star}_i \text{ for all }i,
\end{equation*}
i.e., $Q_i$ is a pointwise lower bound to $Q^{\star}_i$. Due to the minimum operator, $\mathcal{C}_i$ is nonlinear in $Q_i$, although it can be relaxed to the following linear inequality
\begin{equation*}
\begin{aligned}
&Q_i(x_k,r_k,a_{ik},a_{-ik}) \leq l_i(x_k,r_k,a_{ik},a_{-ik}) \\
&+ \gamma Q_i\big(x_{k+1},r_{k+1},u_i,\mu_{-i}(x_{k+1},r_{k+1})\big)
\end{aligned}  
\end{equation*}
for all $(x_k,r_k,a_{ik},a_{-ik},u_i)\in \mathcal{Z}_i\times \mathcal{U}_i$. This relaxation leads to the formulation of the following infinite-dimensional linear program for all $i$
\begin{equation}
\begin{aligned}
\label{eq:27}
\max_{Q_i \in \mathcal{F}(\mathcal{Z}_i)} \quad & \int_{\mathcal{Z}_i}Q_{i}(x,r,a_{i},a_{-i})c_i(dx,dr,da_{i},da_{-i})\\
\textrm{s.t.} \quad \quad & Q_{i}(x_k,r_k,a_{ik},a_{-ik}) \leq l_{i}(x_k,r_k,a_{ik},a_{-ik})\\
\quad & + \gamma Q_{i}\big(x_{k+1},r_{k+1},u_{i},\mu_{-i}(x_{k+1},r_{k+1})\big)\\
\quad & \forall (x_k,r_k,a_{ik},a_{-ik},u_{i})\in \mathcal{Z}_{i}\times \mathcal{U}_{i},
\end{aligned}
\end{equation}
where $c_i$ is a probability measure that allocates positive mass to all open subsets of $\mathcal{Z}_{i}$, for all $i$ \cite{b44,b45,b46,b47,b48}.
\begin{proposition}
For all $i\in \mathcal{P}$, let $Q^{\star}_{i} \in \mathcal{F}(\mathcal{Z}_i)$. Then, the solution to \eqref{eq:8} coincides with a solution to the linear program \eqref{eq:27}, for $c_i$ almost all $(x_k,r_k,a_{ik},a_{-ik})\in \mathcal{Z}_i$.
\end{proposition}
The proof is similar to \cite[Prop. 1]{b49} and is omitted. As a direct consequence of Proposition 2, if $Q^{\star}_{i} \in \mathcal{F}(\mathcal{Z}_i)$ and $\{\mu_{-i}\}=\{\mu^{\star}_{-i}\}$ for all $i$, then the solution to \eqref{eq:11}
coincides with the solution to \eqref{eq:27} for $c_i$ almost all $(x_k,r_k,a_{ik},a_{-ik})\in \mathcal{Z}_i$. Furthermore, the set of control policies $\{\mu^{\star}_i\}_{i=1}^{N}$ constitutes a NE solution based on Theorem 1, with $\mu^{\star}_i$ given by \eqref{eq:12} for all $i$. On a further note, the equivalence of solutions requires that there exists a $\tilde{Q}_i \in \mathcal{F}(\mathcal{Z}_i)$ for which $\tilde{Q}_i(x_k,r_k,a_{ik},a_{-ik})\leq \mathcal{C}_i\tilde{Q}_i(x_k,r_k,a_{ik},a_{-ik})$ is satisfied with equality for all $(x_k,r_k,a_{ik},a_{-ik})\in \mathcal{Z}_{i}$ and $i\in \mathcal{P}$. \\
The computation of an optimizer for \eqref{eq:27} is generally intractable \cite{b48}. To tackle this challenge, we employ the critic-only approximation scheme presented in Section IV. We consider again a restricted function space spanned by a finite number of linearly independent polynomial basis functions $\hat{\mathcal{F}}(\mathcal{Z}_i)=\{\hat{Q}_{i}(\cdot,\cdot,\cdot,\cdot)|\hat{Q}_{i}(x,r,a_{i},a_{-i}) =[\Phi_i(x,r,a_{i},a_{-i})]^{T}\hat{w}_{i}\}$ with $\Phi_i \in \mathbb{R}^{K}$ and $\hat{w}_{i}\in \mathbb{R}^{K}$. This approximation approach leads to the associated control policies $\{\hat{\mu}_i\}_{i=1}^{N}$. Then, an approximate solution to \eqref{eq:27} can be computed by solving the following linear program for all $i\in \mathcal{P}$
\begin{equation}
\begin{aligned}
\label{eq:28}
\max_{\hat{Q}_i \in \hat{\mathcal{F}}(\mathcal{Z}_i)} \quad & \int_{\mathcal{Z}_i}\hat{Q}_{i}(x,r,a_{i},a_{-i})c_i(dx,dr,da_{i},da_{-i})\\
\textrm{s.t.} \quad \quad  & \hat{Q}_{i}(x_k,r_k,a_{ik},a_{-ik}) \leq l_{i}(x_k,r_k,a_{ik},a_{-ik})\\
\quad & + \gamma \hat{Q}_{i}\big(x_{k+1},r_{k+1},u_{i},\hat{\mu}_{-i}(x_{k+1},r_{k+1})\big)\\
\quad & \forall (x_k,r_k,a_{ik},a_{-ik},u_{i})\in \mathcal{Z}_{i}\times \mathcal{U}_{i}.
\end{aligned}
\end{equation}
\textit{Remark 3:} We note that the approximation quality of a solution to \eqref{eq:28} in general depends on the choice of $c_i$ for all $i$ \cite{b44,b45,b46,b47,b48}. However, based on Proposition 2, if $Q^{\star}_i \in \hat{\mathcal{F}}(\mathcal{Z}_i)$, then \eqref{eq:28} does retrieve $Q^{\star}_i$, as long as $c_i$ assigns positive mass to all open subsets of $\mathcal{Z}_i$ for all $i$. \\
We then derive a data-driven implementation by collecting system data. Similar to \eqref{eq:21}, a game data buffer $S^{p}_i$ is constructed by all players $i\in \mathcal{P}$ during a game round $p$. Based on the constructed buffer, we can therefore replace the inequality constraints in \eqref{eq:28} with their sampled variants, leading to a tractable finite-dimensional linear program for the policy evaluation of MSQVI \eqref{eq:16}, for all $i$ 
\begin{align}
\label{eq:29}
\max_{\hat{Q}^{p+1}_i \in \hat{\mathcal{F}}(\mathcal{Z}_i)} \quad & \int_{\mathcal{Z}_i}\hat{Q}^{p+1}_{i}(x,r,a_{i},a_{-i})c_i(dx,dr,da_{i},da_{-i})\nonumber\\
\textrm{s.t.} \quad \quad \quad & \hat{Q}^{p+1}_{i}(x_b,r_b,a_{ib},a_{-ib}) \leq l_{i}(x_b,r_b,a_{ib},a_{-ib})\nonumber\\
\quad & +\sum_{m=1}^{H_p-1}\gamma^{m}l^{\hat{\mu}^{p}_{i},\hat{\mu}^{p}_{-i}}_i(x_{m,b},r_{m,b}) \\
\quad & +\gamma^{H_p} \hat{Q}^{p,\hat{\mu}^{p}_{i},\hat{\mu}^{p}_{-i}}_{i}(x_{H_p,b},r_{H_p,b})\nonumber\\
\quad \quad  & \forall b=1,\ldots,B,\nonumber
\end{align}
where we have used the compact notation
\begin{equation*}
\begin{aligned}
l^{\hat{\mu}^{p}_{i},\hat{\mu}^{p}_{-i}}_i(x,r)=& l_i\big(x,r,\hat{\mu}^{p}_{i}(x,r),\hat{\mu}^{p}_{-i}(x,r)\big)\\
\hat{Q}^{p,\hat{\mu}^{p}_{i},\hat{\mu}^{p}_{-i}}_i(x,r)=& \hat{Q}^{p}_i\big(x,r,\hat{\mu}^{p}_{i}(x,r),\hat{\mu}^{p}_{-i}(x,r)\big)
\end{aligned}
\end{equation*}
to simplify presentation. Algorithm 3 shows the proposed LP algorithm, which we refer to as MSQVI-LP. Based on Remark 3, the LP reformulation \eqref{eq:29} inherits all monotonicity and convergence guarantees of the standard MSQVI algorithm (Algorithm 1) presented in Section III.\\
\begin{algorithm}
        \caption{The proposed MSQVI-LP algorithm.}\label{euclid}
        \begin{algorithmic}[1]
		 \State \textbf{Initialization:} Define $\{\hat{Q}_{i}^{0}\}_{i=1}^{N}=\{Q_{i}^{0}\}_{i=1}^{N}$ based on Remark 1 and arbitrary $\{\hat{\mu}_{i}^{-1}(x,r)\}_{i=1}^{N}=\{\mu_{i}^{-1}(x,r)\}_{i=1}^{N}$. Set $p=0$, $\tau\geq 0$ and $B\in \mathbb{N}$.
               \State \textbf{Policy Improvement:}\newline
 $\hat{\mu}^{p}_{i}(x,r)=\underset{u_{i}}{\mathrm{argmin}}\enspace \hat{Q}^{p}_{i}\big(x,r,u_{i},\hat{\mu}^{p-1}_{i}(x,r)\big)$, for all $i$.
             \State \textbf{Choice of horizon length:} Select $H_p\geq 1$.
             \State  \textbf{Data collection:} Construct data buffer $S^{p}_{i}$ \eqref{eq:21}, for all $i$.
             \State \textbf{Policy Evaluation:} Solve optimization problem \eqref{eq:29}, for all $i$.
\State \textbf{Termination of Learning Phase:} \newline
If $\underset{b}{\mathrm{max}} |\hat{Q}^{p+1}_{i}(x_b,r_b,a_{ib},a_{-ib})-\hat{Q}^{p}_{i}(x_b,r_b,a_{ib},a_{-ib})|> \tau$ for $i\in \mathcal{P}$, set $p=p+1$ and go to Step $2$. Otherwise, set $\{\hat{Q}^{\star}_{i}(x,r,a_i,a_{-i})\}_{i=1}^{N}=\{\hat{Q}^{p+1}_{i}(x,r,a_i,a_{-i})\}_{i=1}^{N}$ and return $\{\hat{\mu}^{\star}_i(x,r)\}_{i=1}^{N}= \{\hat{\mu}^{p+1}_i(x,r)\}_{i=1}^{N}$ as the set of approximate optimal control policies.
        \end{algorithmic}
    \end{algorithm}
\textit{Remark 4:} An important advantage of MSQVI-LS and MSQVI-LP (Algorithms 2 and 3 respectively) is that, despite the fact that the introduction of the horizon variable $H_p$ requires the availability of more data, the total number of decision variables, equations and inequality constraints for the solution of \eqref{eq:24} and \eqref{eq:29} only depend on the total size of the game data buffer $B$ and richness of $\hat{\mathcal{F}}(\mathcal{Z}_i)$ for all $i$. In other words, the utilization of $H_p$ does not increase the computational complexity of the derived algorithms.
\section{In Silico Clinical Studies}
In this section, we evaluate the suitability and performance of the proposed MSQVI algorithmic framework on the problem of fully-automated, dual-hormone glucose control of patients diagnosed with T1DM. To accomplish this, the U.S. FDA-accepted DMMS.R simulator (v1.2.1) from the Epsilon Group \cite{b86,b87} has been utilized, which provides a sophisticated simulation environment to test and compare dosing algorithms for personalized, closed-loop DM treatment. The standard adult population provided by the simulator, which consists of 11 virtual subjects, is used to conduct all simulation studies. Furthermore, in the simulator we employ a commercial CGM profile, where sensor readings are provided in $5$-minute measurement intervals as in realistic CGMs \cite{b88}, along with default infusion pump modules. \\
We run 2,000 in silico clinical trials for all virtual subjects. Each in silico trial is characterized by a rich variability profile around a nominal daily meal and exercise scenario. The nominal meal scenario is defined as a set of 6 meals which take place at $[07\text{:}00, 10\text{:}00, 13\text{:}00, 15\text{:}00, 18\text{:}00, 23\text{:}00]$ with carbohydrate (CHO) amounts of $[70, 30, 90, 30, 90, 25]$ grams and a duration of $[30, 15, 45, 15, 45, 20]$ minutes respectively. The nominal exercise scenario starts at $16\text{:}00$, has a moderate intensity and a duration of 30 minutes. The range of the applied variability profile is given as follows: $1)$ $[-60, 60]$ minutes on the meal time, $2)$ $[-40\%, 40\%]$ on the CHO amount, $3)$ $[-50\%, 50\%]$ on meal duration, $4)$ $[-60, 60]$ minutes on exercise time, $5)$ random choice of [light, moderate, intense] as exercise intensity, and $6)$ $[-50\%, 50\%]$ on exercise duration. The introduced variability follows uniform distributions. The resulting randomized meal and exercise scenarios are significantly more challenging compared to real clinical trials reported in the literature \cite{b89,b90}.

\subsection{Game formulation and algorithmic implementation}
We now proceed by formulating the problem as a discrete-time, two-player NZSG. Let $x_k =\begin{bmatrix}x_{1,k}, & x_{2,k}\end{bmatrix}^{T}$ be the state vector, where $x_{1,k}$ is a glucose measurement received by the CGM at time $k$ [mg/dL] and $x_{2,k}$ is the rate of change in blood glucose computed in $30$-minute measurement intervals, i.e., $x_{2,k}=(x_{1,k}-x_{1,k-6})/30$ [mg/dL/min]. We define the glucose reference setpoint $r_k = r^{\star}=120$ mg/dL for all $k$, which is a reliable glycaemic target for effective DM management \cite{b91,b92}. The two players of the NZSG are the control policies $\mu_1(x,r)$ and $\mu_2(x,r)$ associated with the amounts of rapid-acting insulin [U/5mins] and glucagon [mg/5mins] to be administered to the patient respectively. For the stage cost function \eqref{eq:4}, by defining $\mathcal{S}_{11}=1$, $\mathcal{S}_{22}=10^{-3}$, $R_{11}=100$, $R_{12}=R_{21}=100$ and $R_{22}=300$, we get $l_i\big(x_k,r_k,\mu_i(x_k,r_k),\mu_{-i}(x_k,r_k)\big)= (x_{1,k}-r_k)^{T}\mathcal{S}_{ii}(x_{1,k}-r_k)+ \sum_{j=1}^{2} \mu_{j}^{T}(x_k,r_k)R_{ij}\mu_{j}(x_k,r_k)$ for $i\in \{1,2\}$.
The discount factor is set to $\gamma=0.95$. During each in silico trial, $x_{1,0}$ is randomly initialized in the range $[70, 180]$ mg/dL based on a uniform distribution, while the quantity $x_{1,k-6}$ in the definition of $x_{2,k}$ is initially set to 0 until there are available measurements to utilize.   
\definecolor{Gray}{rgb}{0.498,0.498,0.498}
\begin{table*}[t]
\centering
\caption{GLYCAEMIC CONTROL RESULTS OF THE CONDUCTED IN SILICO CLINICAL TRIALS (LEARNING PHASE).}
\definecolor{Silver}{rgb}{0.749,0.749,0.749}
\definecolor{Black}{rgb}{0,0,0}
\scalebox{0.85}{\begin{tblr}{
  cells = {c},
  hlines,
  vlines = {Silver},
  vline{2} = {-}{Black},
}
\textbf{~}                                     & {\textbf{BG}\\\textbf{mean}\\\textbf{[mg/dL]}} & {\textbf{BG}\\\textbf{min}\\\textbf{[mg/dL]}} & {\textbf{BG}\\\textbf{max}\\\textbf{[mg/dL]}} & {\textbf{\% in}\\\textbf{target}\\\textbf{range}} & {\textbf{\% in}\\\textbf{mild}\\\textbf{hypo}} & {\textbf{\% in}\\\textbf{severe}\\\textbf{hypo}} & {\textbf{\% in}\\\textbf{mild }\\\textbf{hyper}} & {\textbf{\% in}\\\textbf{severe}\\\textbf{hyper}} & \textbf{LBGI} & \textbf{HBGI} & {\textbf{TDI}\\\textbf{[U/day]}} & {\textbf{TDG}\\\textbf{[mg/day]}} & {\textbf{iterations}\\\textbf{till}\\\textbf{convergence}} \\
{\textbf{MSQVI}\\\textbf{(LS/LP)}\\\textbf{~}} & 152±9                                          & 85±10                                         & 194±22                                        & 87.2±6.6                                          & 0.6±0.5                                        & 0±0                                              & 12.2±6.1                                         & 0±0                                               & 0.29±0.15     & 1.52±0.55     & 46.6±13.2                        & 0.58±0.25                         & 90±8                                                       \\
{\textbf{VI}\\\textbf{(LS/LP)}}                & 162±17                                         & 64±12                                         & 241±36                                        & 72.6±7.9                                          & 2.5±0.4                                        & 1.3±1.2                                          & 21.6±4.4                                         & 2.0±1.9                                           & 0.96±0.37     & 3.67±0.76     & 42.8±15.8                        & 0.41±0.21                         & 240±20                                                     
\end{tblr}}
\end{table*} 

\begin{table*}[t]
\centering
\caption{GLYCAEMIC CONTROL RESULTS OF THE CONDUCTED IN SILICO CLINICAL TRIALS,\\ USING THE CONVERGED PERSONALIZED SET OF INSULIN AND GLUCAGON CONTROLLERS.}
\definecolor{Silver}{rgb}{0.749,0.749,0.749}
\definecolor{Black}{rgb}{0,0,0}
\scalebox{0.85}{\begin{tblr}{
  cells = {c},
  hlines,
  vlines = {Silver},
  vline{2} = {-}{Black},
}
\textbf{~}                                     & {\textbf{BG}\\\textbf{mean}\\\textbf{[mg/dL]}} & {\textbf{BG}\\\textbf{min}\\\textbf{[mg/dL]}} & {\textbf{BG}\\\textbf{max}\\\textbf{[mg/dL]}} & {\textbf{\% in}\\\textbf{target}\\\textbf{range}} & {\textbf{\% in}\\\textbf{mild}\\\textbf{hypo}} & {\textbf{\% in}\\\textbf{severe}\\\textbf{hypo}} & {\textbf{\% in}\\\textbf{mild}\\\textbf{~hyper}} & {\textbf{\% in}\\\textbf{severe}\\\textbf{hyper}} & \textbf{LBGI} & \textbf{HBGI} & {\textbf{TDI}\\\textbf{[U/day]}} & {\textbf{TDG}\\\textbf{[mg/day]}} \\
{\textbf{MSQVI}\\\textbf{(LS/LP)}\\\textbf{~}} & 139±8                                          & 96±9                                          & 184±19                                        & 93.1±4.4                                          & 0±0                                            & 0±0                                              & 6.9±4.4                                          & 0±0                                               & 0.01±0.02     & 0.71±0.61     & 50.8±11.2                        & 0.51±0.21                         \\
{\textbf{VI}\\\textbf{(LS/LP)}}                & 152±15                                         & 84±12                                         & 229±28                                        & 80.8±5.3                                          & 0.2±0.3                                        & 0±0                                              & 17.5±4.5                                         & 1.5±0.5                                           & 0.06±0.08     & 2.85±1.35     & 45.9±12.4                        & 0.35±0.19                         
\end{tblr}}
\end{table*} 
The Q-function $\hat{Q}^{p}_{i}(x,r,a_{i},a_{-i})$ is defined as the sum of unique elements derived from the polynomial basis function $X^{T}_i\hat{W}^{p}_iX_i$ for all $i\in \{1,2\}$ and $p\in \mathbb{N}_0$. Here, $X_i =\begin{bmatrix} x_1 & x_2 & x^{2}_1 & x^{2}_{2} & r & r^{2} & a_i & a_{-i} \end{bmatrix}^{T} \in \mathbb{R}^{8}$ and $\hat{W}^{p}_i\in \mathbb{S}^{8\times 8}$ is a symmetrix matrix of the unknown weights. This results in a space $\hat{F}(\mathcal{Z}_i)$ spanned by $K=36$ polynomial basis functions, i.e., $\hat{Q}^{p}_{i}(x,r,a_{i},a_{-i}) =[\Phi_i(x,r,a_{i},a_{-i})]^{T}\hat{w}^{p}_{i}$ with $\Phi_i(x,r,a_{i},a_{-i})\in \mathbb{R}^{36}$ and $\hat{w}^{p}_{i} \in \mathbb{R}^{36}$ for all $i$. The relevance weight $c_i$ in the MSQVI-LP algorithm is a probability measure for all $i$. By setting its first moment as $\tilde{f}_{i}=0_{5\times 1}$ for all $i$, the objective function in the LP problem \eqref{eq:29} simplifies to \cite{b46,b47,b48}
\begin{align*}
&\int_{\mathcal{Z}_i}\hat{Q}^{p+1}_{i}(x,r,a_{i},a_{-i})c_i(dx,dr,da_{i},da_{-i}) =\\
& [\hat{q}^{p+1}_{i,1}]^{T}\tilde{c}_i +[\hat{q}^{p+1}_{i,2}]^{T}\tilde{s}_i+[\hat{q}^{p+1}_{i,3}]^{T}\tilde{k}_i,
\end{align*}
where $\hat{q}^{p+1}_{i,1}\in \mathbb{R}^{15}$, $\hat{q}^{p+1}_{i,2}\in \mathbb{R}^{15}$ and $\hat{q}^{p+1}_{i,3}\in \mathbb{R}^{6}$ are elements of the weight vector $\hat{w}^{p+1}_i$ with second, third and fourth moments given by $\tilde{c}_i \in \mathbb{R}^{15}$, $\tilde{s}_i \in \mathbb{R}^{15}$ and $\tilde{k}_i \in \mathbb{R}^{6}$ respectively. Here, we choose $\tilde{c}_i=\tilde{s}_i=1_{15\times 1}$ and $\tilde{k}_i=1_{6\times 1}$ for all $i$. We set the convergence threshold $\tau=10^{-10}$ and the horizon length $H_p=3$ for all $p$, tuned through the in silico studies. In particular, our results suggest that higher values of $H_p$ will not lead to any noticeable glycaemic control improvement. The size of the game data buffer $S^{p}_i$ in \eqref{eq:21} is set to $B_{MS}=48$ for all $i$, which refers to $12$-hour measurement intervals. Both MSQVI-LS and MSQVI-LP algorithms are initialized with $\hat{\mu}^{-1}_1(x,r)=\hat{\mu}^{-1}_2(x,r)=0$, while $\hat{Q}^{0}_i(x,r,a_i,a_{-i})$ is initialized as a sufficiently large, positive definite function for all $i$. We note that the initial weight value associated with the basis function $a^{2}_1$ on $\hat{Q}^{0}_{1}(x,r,a_1,a_2)$ and $a^{2}_2$ on $\hat{Q}^{0}_{2}(x,r,a_2,a_1)$ must hold sufficiently higher values compared to all other weights of the related Q-function (i.e., around $10^{5}\times$ and $10^{8}\times$ higher than all other weight elements on $\hat{Q}^{0}_{1}$ and $\hat{Q}^{0}_{2}$ respectively), so that the initial policy improvement of MSQVI-LS and MSQVI-LP algorithms can compute reasonable insulin and glucagon policies to be used for the patients. The control action $a_{ik}$ is given by
\begin{align*}
a_{ik} = \begin{cases} \hat{\mu}^{p}_i(x_k,r_k)+n_{ik}, \text{ if } p=0,\\
\big(\hat{\mu}^{p}_i(x_k,r_k)+\hat{\mu}^{p-1}_i(x_k,r_k)\big)/2+n_{ik}, \text{ if } p>0,
\end{cases}
\end{align*}
where $n_{ik}$ is a probabilistic sample drawn from a uniform distribution over $[10^{-3},5\cdot 10^{-3}]$ for $i=1$ and $[10^{-5},5\cdot 10^{-5}]$ for $i=2$. Finally, to fully evaluate the capabilities of the proposed MSQVI framework, we compare the performance of the MSQVI-LS and MSQVI-LP algorithms with the LS and LP based VI algorithms (obtained by setting $H_p=1$ for all $p$ in Algorithms 2 and 3), which we refer to as VI-LS and VI-LP respectively. These algorithms utilize the exact same configuration discussed above, with the exception that the size of the data buffer $S^{p}_i$ in \eqref{eq:21} is set to $B_{VI}=144$ for all $i$, to retain the $12$-hour measurement intervals. On a final note, all meal and exercise information is kept completely unannounced to the algorithms, as expected in the design of a truly fully-automated AP system.
\begin{figure*}[h!]
\centering
\begin{multicols}{2}
  \hspace*{-1.2cm}\includegraphics[width=1.1\linewidth]{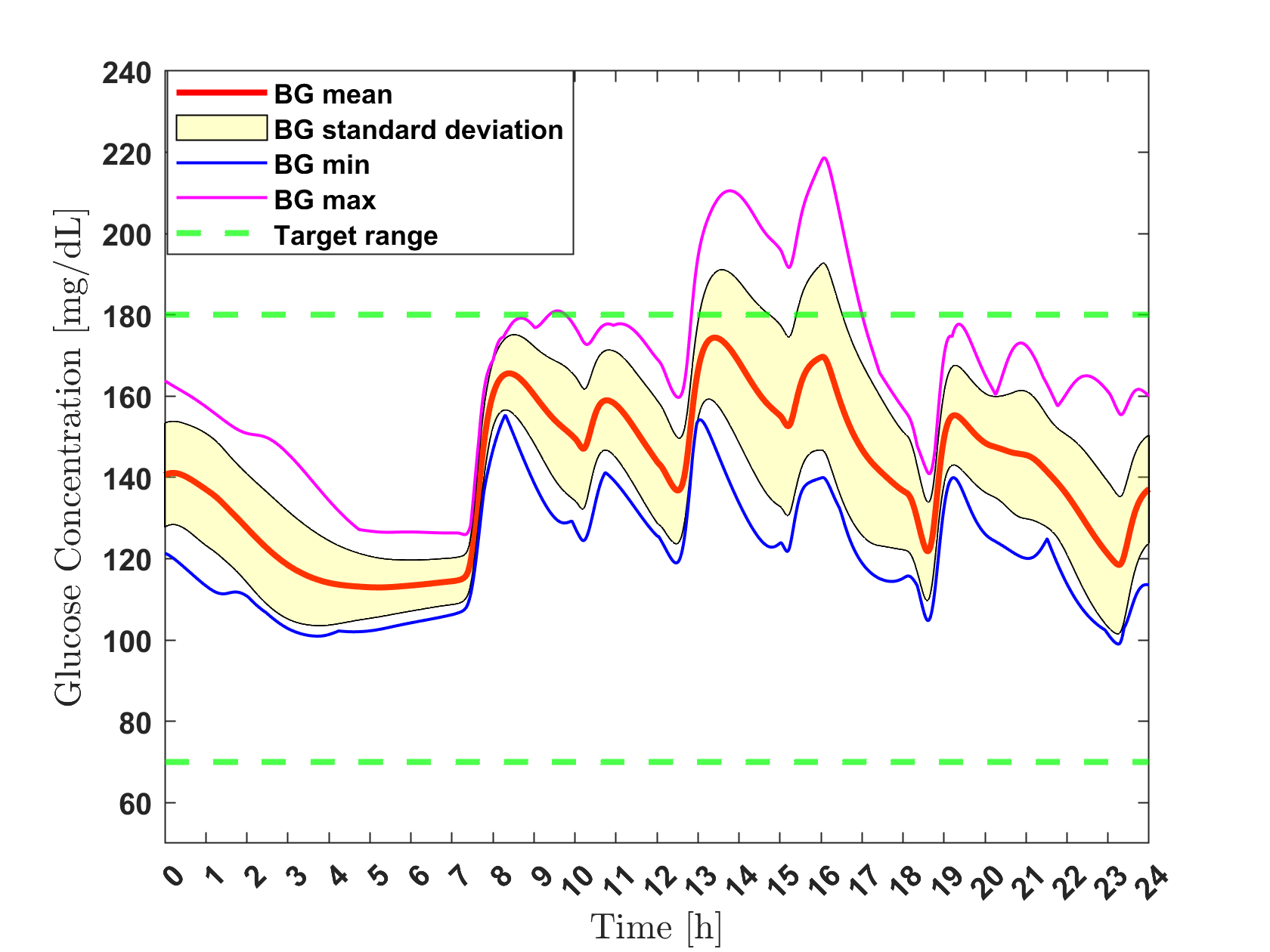}\par 
 \hspace*{-0.4cm} \includegraphics[width=1.1\linewidth]{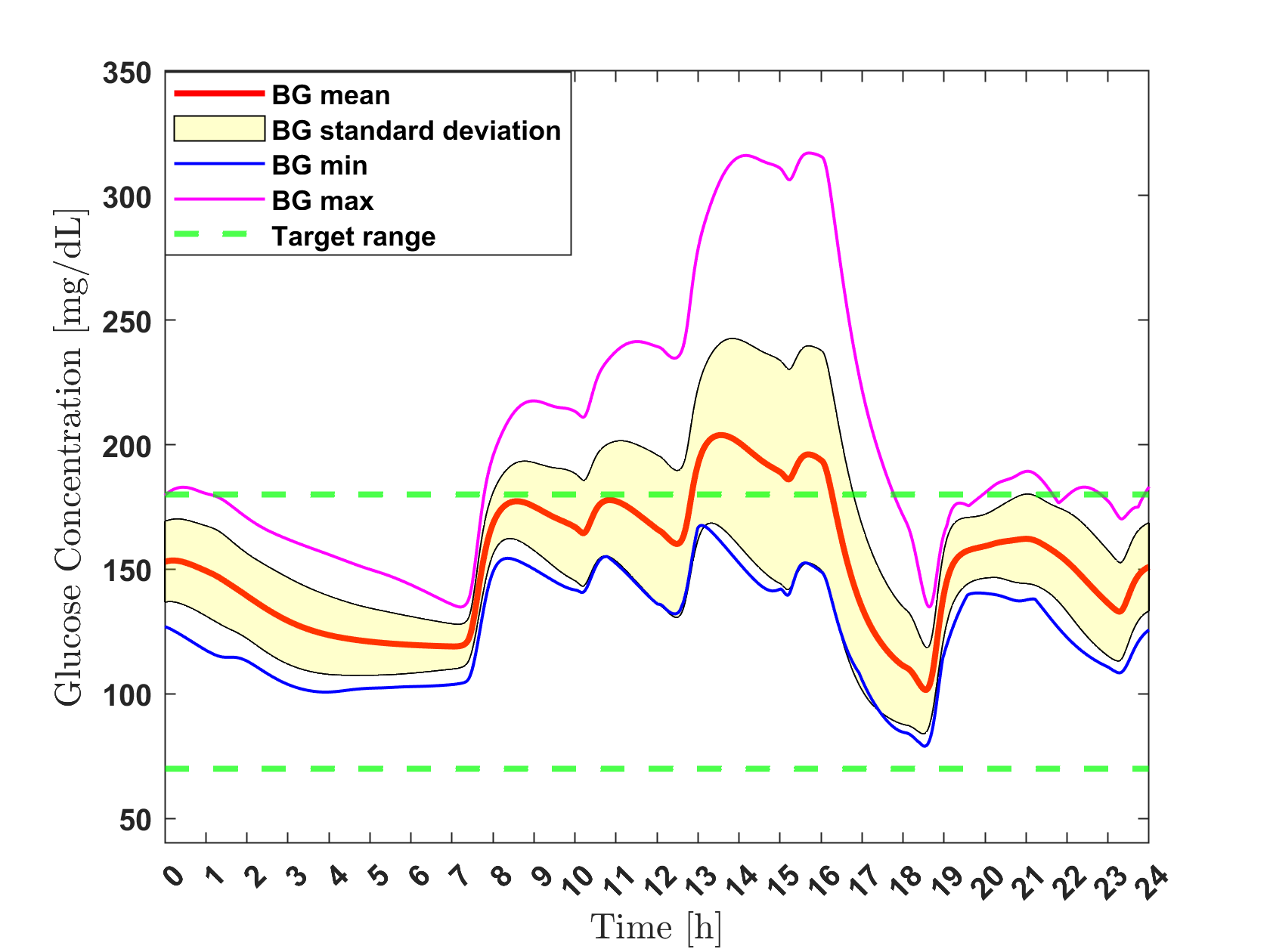}\par 
    \end{multicols}
\caption{Glycaemic control performance of MSQVI-LS/LP (left) and VI-LS/LP (right) algorithms under the converged personalized set of insulin and glucagon controllers, during a single day characterized by the nominal daily meal and exercise configuration.}
\end{figure*}
\subsection{Evaluation Metrics and Results}
We assess a variety of widely used, clinically validated metrics \cite{b93} for the entire virtual adult population:
\begin{itemize}
\item The mean, minimum and maximum values of the blood glucose measurements during the conduction of the in silico studies, 
\item The percentages of time in normoglycaemic (glucose measurements within $[70, 180]$ mg/dL), mild hypoglycaemic (glucose measurements within $[50, 70)$ mg/dL), severe hypoglycaemic (glucose measurements $<50$ mg/dL), mild hyperglycaemic (glucose measurements within $(180, 250]$ mg/dL) and severe hyperglycaemic (glucose measurements $> 250$ mg/dL) ranges,
\item The low and high blood glucose indices (LBGI and HBGI respectively), which provide a measure of the extent and frequency of low and high blood glucose measurements respectively \cite{b94,b95}, and
\item The total daily amounts of insulin (TDI) and glucagon (TDG) delivery, as well as the total number of iterations until convergence of the respective algorithms.
\end{itemize}
The results are reported for the entire adult population in the format [mean value $\pm$ standard deviation]. \\
Table I presents the results of the in silico clinical studies until convergence of the implemented MSQVI and VI algorithms, which we refer to as the Learning Phase. We firstly observe that both LS and LP variants of the MSQVI and VI algorithms produce the same results. This is expected, since both LS and LP variants share the same family of approximate Q-functions and inherit the monotonicity and convergence guarantees of the standard theoretical algorithm (studied in Section III). Furthermore, the MSQVI framework requires significantly fewer iterations to converge compared to VI ($90\pm 8$ vs $240\pm 20$), which translates to signifantly fewer days ($45\pm 4$ days vs $120 \pm 10$ days). Moreover, MSQVI provides outstanding glycaemic control during the Learning Phase, characterized by significantly higher percentages of time in the target range ($87.2 \pm 6.6$ vs $72.6 \pm 7.9$), lower percentages of time in the mild hypoglycaemic ($0.6\pm0.5$ vs $2.5\pm 0.4$) and hyperglycaemic ($12.2\pm6.1$ vs $21.6\pm 4.4$) ranges, with no time spent in severe hypoglycaemic 
(compared to $1.3 \pm 1.2$ of VI) and hyperglycaemic (compared to $2.0\pm 1.9$ for VI) ranges. This leads MSQVI to enable significantly improved mean, minimum and maximum observed blood glucose measurements ($152\pm 9$ vs $162\pm 17$, $85\pm 10$ vs $64\pm 12$ and $194\pm 12$ vs $241\pm 36$ respectively), as well as significantly lower values of LBGI ($0.29\pm 0.15$ vs $0.96\pm 0.37$) and HBGI ($1.52 \pm 0.55$ vs $3.67 \pm 0.76$).\\
After algorithmic convergence, we repeat all in silico trials for a duration of 60 days, by employing now the converged personalized set of insulin and glucagon controllers associated with each virtual adult subject in the population. Table II shows the related results. As expected, the approximate optimal insulin and glucagon policies previously computed by the MSQVI and VI algorithmic variants achieve better glycaemic behavior compared to the Learning Phase. However, MSQVI again accomplishes crucially better glycaemic control compared to VI, with significantly higher time spent in the target range ($93.1\pm 4.4$ vs $80.8 \pm 5.3$) and less time spent in mild hyperglyceamia ($6.9 \pm 4.4$ vs $17.5\pm 4.5$). Furthermore, MSQVI leads to no time spent in mild hypoglyceamic (compared to $0.2\pm 0.3$ of VI), severe hypoglycaemic and severe hyperglycaemic (compared to $1.5\pm 0.5$ of VI) ranges. Similar significant improvements are observed in the values of LBGI, HBGI and blood glucose related metrics.\\
For demonstration purposes, Figure 1 shows the population-based glycaemic responses achieved by the converged, personalized insulin and glucagon policies, as computed by the MSQVI and VI algorithms, during a single day of a conducted in silico trial and under the nominal daily meal and exercise scenario defined above. It is clear that the MSQVI framework enables better glycaemic regulation under the presence of completely unannounced meals and exercises. Overall, the proposed MSQVI algorithmic framework not only enables higher convergence speed, but also better and clinically smoother control solutions (in terms of glycaemic control performance) \cite{b96}, as expected based on the theoretical findings of Sections III, IV and V.
\section{Conclusion}
In this work, we proposed novel model-free RL/ADP algorithms, with critical applications to the design of fully-automated, closed-loop drug delivery systems for personalized medicine. We derived a novel, theoretically rigorous Q-function-based MSQVI algorithm for optimal tracking control of unknown discrete-time NZSGs. The proposed algorithmic framework integrates the complimentary strengths of classical single-step PI and VI algorithms, i.e., fast convergence to approximate optimal solutions with an easy-to-realize initialization condition. A critic-only LS implementation approach was then developed, significantly reducing the total computational burden compared to conventional multiple NN approximation methods. Afterwards, a novel LP approach for unknown discrete-time NZSGs is derived, successfully extending the optimization framework to the critical context of discrete-time, game theoretical control problems. The high performance and reliability of the proposed algorithms are evaluated in simulation, on the challenging problem of fully-automated, dual-hormone glucose control in T1DM, by utilizing a U.S. FDA-accepted metabolic simulator.\\
As a future work, we aim to extend the derived MSQVI algorithm to the critical setting of unknown NZSGs under the existence of dynamic uncertainty. This can be achieved by integrating novel robust stabilization methods that can ensure the stability of the closed-loop multiplayer system. Furthermore, we plan to implement the proposed algorithms on real embedded hardware. This will enable the design of a wearable, dual-hormone artificial pancreas prototype that can be potentially employed in future clinical trials for real T1DM patients.

\appendices
\section{Proof of Theorem 2}
To simplify presentation of the proof, we define the following compact notation
\begin{align}
\label{eq:30}
F_i^{\mu_{i},\mu_{-i}}(x,r) =& F_i\big(x,r,\mu_{i}(x,r),\mu_{-i}(x,r)\big) \nonumber \\
F_i^{p,\mu_{i},\mu_{-i}}(x,r) =& F^{p}_i\big(x,r,\mu_{i}(x,r),\mu_{-i}(x,r)\big)
\end{align}
for generic functions $F_i:\mathcal{Z}_i\rightarrow \mathbb{R}_{+}$ and $F^{p}_i:\mathcal{Z}_i\rightarrow \mathbb{R}_{+}$, for all $i\in \mathcal{P}$ and $p\in \mathbb{N}_0$. \\
1) We apply mathematical induction to prove \eqref{eq:18}. Based on \eqref{eq:15} and \eqref{eq:16}, we get
\begin{align*}
Q^{1}_{i}(x,r,a_{i},a_{-i})=& l_{i}(x,r,a_{i},a_{-i})+\sum_{m=1}^{H_0-1}\gamma^{m}l_{i}^{\mu^{0}_{i},\mu^{0}_{-i}}(x_m,r_m)\\
&+\gamma^{H_0} Q^{0,\mu^{0}_i,\mu^{0}_{-i}}_i(x_{H_0},r_{H_0})  \\
=& l_{i}(x,r,a_{i},a_{-i})+\sum_{m=1}^{H_0-2}\gamma^{m}l_{i}^{\mu^{0}_{i},\mu^{0}_{-i}}(x_m,r_m) \\
&+\gamma^{H_0-1} l_{i}^{\mu^{0}_{i},\mu^{0}_{-i}}(x_{H_0-1},r_{H_0-1})  \\
&+\gamma^{H_0} Q^{0,\mu^{0}_i,\mu^{0}_{-i}}_i(x_{H_0},r_{H_0}) \\
= &l_{i}(x,r,a_{i},a_{-i})+\sum_{m=1}^{H_0-2}\gamma^{m}l_{i}^{\mu^{p}_{i},\mu^{p}_{-i}}(x_m,r_m) \\
&+\gamma^{H_0-1}\bigg[l_{i}^{\mu^{0}_{i},\mu^{0}_{-i}}(x_{H_0-1},r_{H_0-1}) \\
&+ \gamma Q^{0,\mu^{0}_i,\mu^{0}_{-i}}_i(x_{H_0},r_{H_0})\bigg] \\
\overset{\eqref{eq:17}}{\leq}& l_{i}(x,r,a_{i},a_{-i})+\sum_{m=1}^{H_0-2}\gamma^{m}l_{i}^{\mu^{0}_{i},\mu^{0}_{-i}}(x_m,r_m)  \\
&+\gamma^{H_{0}-1}Q^{0,\mu^{0}_i,\mu^{0}_{-i}}_i(x_{H_0-1},r_{H_0-1}).
\end{align*}
Iterating leads to
\begin{equation}
\label{eq:31}
Q^{1}_{i}(x,r,a_{i},a_{-i}) \leq l_{i}(x,r,a_{i},a_{-i})+\gamma Q^{0,\mu^{0}_i,\mu^{0}_{-i}}_i(x_1,r_1).
\end{equation}
Hence, by \eqref{eq:17},
\begin{align*}
Q^{1}_{i}(x,r,a_{i},a_{-i})\leq & l_{i}(x,r,a_{i},a_{-i})+\gamma Q^{0,\mu^{0}_i,\mu^{0}_{-i}}_i(x_1,r_1) \nonumber \\
\leq & Q^{0}_{i}(x,r,a_{i},a_{-i}).
\end{align*}
Then, we assume that \eqref{eq:18} holds for $p-1$,
\begin{align}
\label{eq:32}
Q^{p}_{i}(x,r,a_{i},a_{-i})\leq & l_{i}(x,r,a_{i},a_{-i})+\gamma Q^{p-1,\mu^{p-1}_{i},\mu^{p-1}_{-i}}_{i}(x_{1},r_1) \nonumber \\
\leq & Q^{p-1}_{i}(x,r,a_{i},a_{-i}).
\end{align}
It follows that
\begin{align}
\label{eq:33}
Q^{p}_{i}(x,r,a_{i},a_{-i})=& l_{i}(x,r,a_{i},a_{-i})\nonumber \\
&+\sum_{m=1}^{H_{p-1}-1}\gamma^{m}l_{i}^{\mu^{p-1}_{i},\mu^{p-1}_{-i}}(x_m,r_m)\nonumber \\
& +\gamma^{H_{p-1}} Q^{p-1,\mu^{p-1}_i,\mu^{p-1}_{-i}}_i(x_{H_{p-1}},r_{H_{p-1}})\nonumber  \\
\overset{\eqref{eq:32}}{\geq} & l_{i}(x,r,a_{i},a_{-i}) \nonumber \\
&+\sum_{m=1}^{H_{p-1}-1}\gamma^{m}l_{i}^{\mu^{p-1}_{i},\mu^{p-1}_{-i}}(x_m,r_m)\nonumber  \\
&+\gamma^{H_{p-1}}\bigg[l_{i}^{\mu^{p-1}_{i},\mu^{p-1}_{-i}}(x_{H_{p-1}},r_{H_{p-1}})\nonumber \\
&+\gamma Q^{p-1,\mu^{p-1}_i,\mu^{p-1}_{-i}}_i(x_{H_{p-1}+1},r_{H_{p-1}+1})\bigg]\nonumber \\
=& l_{i}(x,r,a_{i},a_{-i})\nonumber \\
&+\sum_{m=1}^{H_{p-1}}\gamma^{m}l_{i}^{\mu^{p-1}_{i},\mu^{p-1}_{-i}}(x_m,r_m)\nonumber \\
&+ \gamma^{H_{p-1}+1}\bigg[ \nonumber \\
& Q^{p-1,\mu^{p-1}_i,\mu^{p-1}_{-i}}_i(x_{H_{p-1}+1},r_{H_{p-1}+1})\bigg]\nonumber \\
=&l_{i}(x,r,a_{i},a_{-i}) \nonumber \\
&+\gamma \bigg[\sum_{m=1}^{H_{p-1}}\gamma^{m-1}l_{i}^{\mu^{p-1}_{i},\mu^{p-1}_{-i}}(x_m,r_m) \nonumber  \\
&+\gamma^{H_{p-1}}\cdot \nonumber \\
& Q^{p-1,\mu^{p-1}_i,\mu^{p-1}_{-i}}_i(x_{H_{p-1}+1},r_{H_{p-1}+1})\bigg]\nonumber \\
\overset{\eqref{eq:16}}{=} &l_{i}(x,r,a_{i},a_{-i})+\gamma Q^{p,\mu^{p-1}_i,\mu^{p-1}_{-i}}_i(x_1,r_1)\nonumber \\
\overset{\eqref{eq:15}}\geq & l_{i}(x,r,a_{i},a_{-i})+\gamma Q^{p,\mu^{p}_i,\mu^{p-1}_{-i}}_i(x_1,r_1).
\end{align}
To proceed, we will now prove that
\begin{equation}
\label{eq:34}
Q^{p,\mu^{p}_{i},\mu^{p-1}_{-i}}_{i}\big(x,r) \geq Q^{p,\mu^{p}_{i},\mu^{p}_{-i}}_{i}\big(x,r)\enspace \forall i\in \mathcal{P}.
\end{equation}
If $g_j(x)=0_{n\times m_j}$ for all $j\in \mathcal{P}\setminus \{i\}$, then considering \eqref{eq:16} we have that
\begin{align}
\label{eq:35}
Q^{p,\mu^{p}_{i},\mu^{p-1}_{-i}}_{i}(x,r)=& l^{\mu^{p}_i,\mu^{p-1}_{-i}}_i(x,r) \nonumber \\
&+\sum_{m=1}^{H_{p-1}-1}\gamma^{m}l^{\mu^{p-1}_i,\mu^{p-1}_{-i}}_i(x_m,r_m) \nonumber \\
& +\gamma^{H_{p-1}} Q^{p-1,\mu^{p-1}_i,\mu^{p-1}_{-i}}_{i}(x_{H_{p-1}},r_{H_{p-1}}) \nonumber \\
=& Q^{p,\mu^{p}_{i},\mu^{p}_{-i}}_{i}(x,r) + E_{1}
\end{align}
where 
\begin{align*}
E_{1} =& \sum_{j=1,j\neq i}^{N}\bigg[\big[\Delta \mu^{p}_j(x,r)\big]^{T} R_{ij}\Delta \mu^{p}_j(x,r) \\
&-2\big[\Delta \mu^{p}_j(x,r)\big]^{T}R_{ij}\mu^{p}_{j}(x,r)\bigg]
\end{align*}
and $\Delta \mu^{p}_j(x,r) = \mu^{p}_j(x,r)-\mu^{p-1}_j(x,r)$. By continuity of $Q^{p}_{i}$ and $l_i$ on $\mathcal{Z}_i$, \eqref{eq:35} also holds for sufficiently small values of $\|g_j\|_2$ for $j\in \mathcal{P}\setminus \{i\}$. Since $Q^{p}_{i}\geq 0$, it suffices to show that $E_{1}\geq 0$ in order for \eqref{eq:34} to hold. Based on \eqref{eq:15},
\begin{align}
\label{eq:36}
\mu^{p}_{j}(x,r) =& \underset{u_{j}}{\mathrm{argmin}}\enspace Q^{p,u_{j},\mu^{p-1}_{-j}}_{j}(x,r) \nonumber\\
=& -\frac{1}{2} R^{-1}_{jj}g^{T}_{j}(x)E_2
\end{align}
where
\begin{align*}
&E_2 = \gamma\frac{\partial l^{\mu^{p-1}_j,\mu^{p-1}_{-j}}_{j}(x_1,r_1)}{\partial x_1}+\sum_{m=2}^{H_{p-1}-1}\bigg[\gamma^{m}\prod_{n=m,m-1,\ldots}^{2}\frac{\partial x_{n}}{\partial x_{n-1}}\\
&\cdot \frac{\partial l^{\mu^{p-1}_j,\mu^{p-1}_{-j}}_{j}(x_m,r_m)}{\partial x_m }\bigg] + \bigg[\gamma^{H_{p-1}}\prod_{n=H_{p-1},H_{p-1}-1,\ldots}^{2} \frac{\partial x_{n}}{\partial x_{n-1}}\\
&\cdot \frac{\partial Q^{p-1,\mu^{p-1}_j,\mu^{p-1}_{-j}}_{j}( x_{H_{p-1}}, r_{H_{p-1}})}{\partial x_{H_{p-1}}}\bigg].
\end{align*}
By using \eqref{eq:36}, $E_1\geq 0$ means
\begin{align}
\label{eq:37}
&\sum_{j=1,j\neq i}^{N}\bigg[\big[\Delta \mu^{p}_{j}(x,r)\big]^{T}R_{ij}\Delta \mu^{p}_{j}(x,r) \nonumber \\
&+\big[\Delta \mu^{p}_{j}(x,r)\big]^{T}R_{ij}R^{-1}_{jj}g^{T}_{j}(x)E_2 \bigg] \geq 0.
\end{align}
By using the sufficient condition
\begin{align*}
\big[\Delta \mu^{p}_{j}(x,r)\big]^{T}R_{ij}\Delta \mu^{p}_{j}(x,r) \geq & \big[\Delta \mu^{p}_{j}(x,r)\big]^{T}R_{ij}R^{-1}_{jj}\\
& \cdot g^{T}_{j}(x)E_2
\end{align*}
for all $j\in \mathcal{P}\setminus \{i\}$, then by using norm properties on \eqref{eq:37} yields
\begin{align}
\label{eq:38}
\sum_{j=1,j\neq i}^{N}\underline{\sigma}(R_{ij})\|\Delta \mu^{p}_{j}(x,r))\|_2 \geq & \sum_{j=1,j\neq i}^{N}\overline{\sigma}(R_{ij}R^{-1}_{jj})\nonumber \\
&\cdot \|g^{T}_j(x)\|_2 \|E_2\|_2.
\end{align}
Assuming that $\overline{\sigma}(R_{ij}R^{-1}_{jj})$ is sufficiently small, condition \eqref{eq:38} holds, and therefore \eqref{eq:34} also holds. Finally, the reasoning leading up to \eqref{eq:31} gives
\begin{align}
\label{eq:39}
Q^{p+1}_{i}(x,r,a_{i},a_{-i})=& l_{i}(x,r,a_{i},a_{-i}) \nonumber\\
&+\sum_{m=1}^{H_p-1}\gamma^{m}l^{\mu^{p}_{i},\mu^{p}_{-i}}_{i}(x_m,r_m)\nonumber \\
&+\gamma^{H_p} Q^{p,\mu^{p}_{i},\mu^{p}_{-i}}_{i}(x_{H_p},r_{H_p}) \nonumber \\
\leq & l_{i}(x,r,a_{i},a_{-i}) \nonumber \\
&+\sum_{m=1}^{H_p-2}\gamma^{m}l^{\mu^{p}_{i},\mu^{p}_{-i}}_{i}\big(x_m,r_m) \nonumber \\
&+\gamma^{H_{p}-1}Q^{p,\mu^{p}_{i},\mu^{p}_{-i}}_{i}\big(x_{H_{p}-1},x_{H_{p}-1})\nonumber \\
\leq & l_{i}(x,r,a_{i},a_{-i}) +\gamma Q^{p,\mu^{p}_{i},\mu^{p}_{-i}}_{i}\big(x_{1},r_{1}).
\end{align}
Therefore, by considering \eqref{eq:33}, \eqref{eq:34} and \eqref{eq:39}, \eqref{eq:18} holds for all $i\in \mathcal{P}$ and $p\geq 0$.\\
2) According to \eqref{eq:18}, the sequence $\{Q^{p}_{i}(x,r,a_{i},a_{-i})\}_{p\in \mathbb{N}_0}$ is non-increasing. Furthermore, since $l_{i}$ is non-negative, the sequence is additionally lower bounded by $0$ for all $p$. Hence, it has a point-wise limit $Q^{\infty}_{i}(x,r,a_{i},a_{-i}) = \lim_{p\rightarrow \infty} Q^{p}_{i}(x,r,a_{i},a_{-i})$. If we define $\mu^{\infty}_{i}(x,r) = \underset{u_{i}}{\mathrm{argmin}}\enspace Q^{\infty,u_{i},\mu^{\infty}_{-i}}_{i}(x,r)$ and take the limit of \eqref{eq:18}, we have that
\begin{align*}
Q^{\infty}_{i}(x,r,a_{i}.a_{-i}) \leq & l_{i}(x,r,a_{i},a_{-i})+\gamma Q^{\infty,\mu^{\infty}_{i},\mu^{\infty}_{-i}}_{i}\big(x_1,r_1) \\
\leq & Q^{\infty}_{i}(x,r,a_{i},a_{-i})
\end{align*}
which leads to
\begin{equation}
\label{eq:40}
Q^{\infty}_{i}(x,r,a_{i}.a_{-i}) = l_{i}(x,r,a_{i},a_{-i})+\gamma Q^{\infty,\mu^{\infty}_{i},\mu^{\infty}_{-i}}_{i}\big(x_1,r_1).
\end{equation}
Due to the uniqueness of solutions to the Bellman equation \cite{b1,b8,b47,b48}, we note that \eqref{eq:40} is essentially \eqref{eq:11}, which means that $Q^{\infty}_{i}(x,r,a_{i},a_{-i})=Q^{\star}_{i}(x,r,a_{i},a_{-i})$ for all $i$, and therefore $\{\mu^{\infty}_{i}(x,r)\}_{i=1}^{N}$ =$\{\mu^{\star}_{i}(x,r)\}_{i=1}^{N}$.
\section{Proof of Corollary 1}
We use the compact notation \eqref{eq:30} to prove the desired statement. Similar to the reasoning for the proof of \eqref{eq:31}, we get
\begin{align*}
Q_{i,h_1}(x,r,a_{i},a_{-i})=&l_{i}(x,r,a_{i},a_{-i})+\sum_{m=1}^{h_1-1}\gamma^{m}l^{\tilde{\mu}_{i},\tilde{\mu}_{-i}}_{i}\big(x_m,r_m) \\
&+\gamma^{h_1} \tilde{Q}^{\tilde{\mu}_i,\tilde{\mu}_{-i}}_i(x_{h_1},r_{h_1}) \\
=& l_{i}(x,r,a_{i},a_{-i}) \\
&+\sum_{m=1}^{h_1-2}\gamma^{m}l^{\tilde{\mu}_{i},\tilde{\mu}_{-i}}_{i}\big(x_m,r_m) \\
&+\gamma^{h_1-1}\big[l_{i}^{\tilde{\mu}_{i},\tilde{\mu}_{-i}}(x_{h_1-1},r_{h_1-1}) \\
&+ \gamma \tilde{Q}^{\tilde{\mu}_i,\tilde{\mu}_{-i}}_i(x_{h_1},r_{h_1})\big]\\
\overset{\eqref{eq:17}}{\leq}& l_{i}(x,r,a_{i},a_{-i})\nonumber \\
&+\sum_{m=1}^{h_1-2}\gamma^{m}l^{\tilde{\mu}_{i},\tilde{\mu}_{-i}}_{i}(x_m,r_m) \\
&+\gamma^{h_1-1}\tilde{Q}^{\tilde{\mu}_i,\tilde{\mu}_{-i}}_i(x_{h_1-1},r_{h_1-1})
\end{align*}
Iterating leads to
\begin{align*}
Q_{i,h_1}(x,r,a_{i},a_{-i}) \leq & l_{i}(x,r,a_{i},a_{-i})\nonumber \\
&+\sum_{m=1}^{h_2-1}\gamma^{m}l^{\tilde{\mu}_{i},\tilde{\mu}_{-i}}_{i}\big(x_m,r_m)\\
&+\gamma^{h_2}\tilde{Q}^{\tilde{\mu}_i,\tilde{\mu}_{-i}}_i(x_{h_2},r_{h_2})\\
= &Q_{i,h_2}(x,r,a_{i},a_{-i}).
\end{align*}
\section{Proof of Theorem 3}
We utilize the compact notation \eqref{eq:30} to prove the desired statements. For all $i\in \mathcal{P}$, let $\bar{Q}^{0}_{i}(x,r,a_{i},a_{-i}) =\hat{Q}^{0}_{i}(x,r,a_{i},a_{-i})=Q^{0}_{i}(x,r,a_{i},a_{-i})$ and $\bar{Q}^{p+1}_{i}(x,r,a_{i},a_{-i})$ satisfies the following equation
\begin{align}
\label{eq:41}
&\bar{Q}^{p+1}_{i}(x,r,a_{i},a_{-i})= l_{i}(x,r,a_{i},a_{-i})\nonumber \\
&+\sum_{m=1}^{H_p-1}\gamma^{m}l^{\hat{\mu}^{p}_{i},\hat{\mu}^{p}_{-i}}_{i}(x_m,r_m)+\gamma^{H_p} \bar{Q}^{p,\hat{\mu}^{p}_{i},\hat{\mu}^{p}_{-i}}_{i}(x_{H_p},r_{H_p}).
\end{align}
Similar to \eqref{eq:19}, $\bar{Q}^{p}_{i}(x,r,a_{i},a_{-i})$ can be expressed as
\begin{equation}
\label{eq:42}
\bar{Q}^{p}_{i}(x,r,a_{i},a_{-i}) = [\Phi_{i}(x,r,a_{i},a_{-i})]^{T}\bar{w}^{p}_{i} + \bar{e}^{p}_{i}(x,r,a_{i},a_{-i}),
\end{equation}
where $\bar{w}^{p}_{i} \in \mathbb{R}^{K}$ and $\bar{e}^{p}_{i}(x,r,a_{i},a_{-i})\in \mathbb{R}$ is the approximation error that satisfies $\lim_{K\to \infty} \bar{e}^{p}_{i}(x,r,a_{i},a_{-i}) = 0$. With \eqref{eq:41} and \eqref{eq:42}, we have that
\begin{align}
\label{eq:43}
&\bar{\epsilon}^{p+1}_{i}(x,r,a_{i},a_{-i})- [\Phi_{i}(x,r,a_{i},a_{-i})]^{T}\bar{w}^{p+1}_{i} \nonumber \\
&+\gamma^{H_p}[\Phi^{\hat{\mu}^{p}_{i},\hat{\mu}^{p}_{-i}}_{i}(x_{H_p},r_{H_p})]^{T}\bar{w}^{p}_{i}+ \sum_{m=1}^{H_{p}-1}\gamma^{m}l^{\hat{\mu}^{p}_{i},\hat{\mu}^{p}_{-i}}_{i}\big(x_{m},r_{m}) \nonumber \\ &+l_{i}(x,r,a_{i},a_{-i})=0,
\end{align} 
where $\bar{\epsilon}^{p+1}_{i}(x,r,a_{i},a_{-i})= \gamma^{H_p}\bar{e}^{p,\hat{\mu}^{p}_{i},\hat{\mu}^{p}_{-i}}_{i}(x_{H_p},r_{H_p})-\bar{e}^{p+1}_{i}(x,r,a_{i},a_{-i})$. Therefore, $\lim_{K \to \infty}\bar{\epsilon}^{p+1}_{i}(x,r,a_{i},a_{-i}) = 0$. By defining $\tilde{w}^{p}_{i} = \hat{w}^{p}_{i}-\bar{w}^{p}_{i}$, \eqref{eq:43} becomes
\begin{align}
\label{eq:44}
&\bar{\epsilon}^{p+1}_{i}(x,r,a_{i},a_{-i}) = [\Phi_{i}(x,r,a_{i},a_{-i})]^{T}(\hat{w}^{p+1}_{i}-\tilde{w}^{p+1}_{i})\nonumber\\
&- \gamma^{H_p}[\Phi^{\hat{\mu}^{p}_{i},\hat{\mu}^{p}_{-i}}_{i}(x_{H_p},r_{H_p})]^{T}(\hat{w}^{p}_{i}-\tilde{w}^{p}_{i})\nonumber \\
&- \sum_{m=1}^{H_{p}-1}\gamma^{m}l^{\hat{\mu}^{p}_{i},\hat{\mu}^{p}_{-i}}_{i}(x_{m},r_m) -l_{i}(x,r,a_{i},a_{-i})\nonumber \\
&= \epsilon^{p+1}_{i}(x,r,a_{i},a_{-i})-[\Phi_{i}(x,r,a_{i},a_{-i})]^{T}\tilde{w}^{p+1}_{i}\nonumber\\
&+ \gamma^{H_p}[\Phi^{\hat{\mu}^{p}_{i},\hat{\mu}^{p}_{-i}}_{i}(x_{H_p},r_{H_p})]^{T}\tilde{w}^{p}_{i}.
\end{align} 
We now proceed by proving that $\lim_{K \to \infty}\tilde{w}^{p}_{i}=0$ using mathematical induction. From the initial condition $\bar{Q}^{0}_{i}(x,r,a_{i},a_{-i}) = \hat{Q}^{0}_{i}(x,r,a_{i},a_{-i})=Q^{0}_{i}(x,r,a_{i},a_{-i})$, we have that $\lim_{K\to \infty} \tilde{w}^{0}_{i}=0$. Assume that $\lim_{K \to \infty} \tilde{w}^{p}_{i}=0$ holds for $p$. Based on \eqref{eq:44}, for each data tuple $b$ in $S^{p}_{i}$, it holds that
\begin{align*}
\bar{\epsilon}^{p+1}_{i}(x_{b},r_b,a_{ib},a_{-ib})=& \epsilon^{p+1}_{i}(x_{b},r_b,a_{ib},a_{-ib})-[\Psi_{i,b}]^{T}\tilde{w}^{p+1}_{i}\nonumber \\
&+ \gamma^{H_p}[\Phi^{\hat{\mu}^{p}_{i},\hat{\mu}^{p}_{-i}}_{i}(x_{H_p,b},r_{H_p,b})]^{T}\tilde{w}^{p}_{i},
\end{align*}
leading to
\begin{align}
\label{eq:45}
[\tilde{w}^{p+1}_{i}]^{T}\Psi_{i,b}[\Psi_{i,b}]^{T}\tilde{w}^{p+1}_{i} =& [\bar{\epsilon}^{p+1}_{i}(x_{b},r_b,a_{ib},a_{-ib})\nonumber \\
&- \epsilon^{p+1}_{i}(x_{b},r_b,a_{ib},a_{-ib})]^{2}+E_{3},
\end{align}
where 
\begin{align*}
E_{3} =& 2\gamma^{H_p}[\epsilon^{p+1}_{i}(x_{b},r_b,a_{ib},a_{-ib})-\bar{\epsilon}^{p+1}_{i}(x_{b},r_b,a_{ib},a_{-ib})]\\
& \cdot [\Phi^{\hat{\mu}^{p}_{i},\hat{\mu}^{p}_{-i}}_{i}(x_{b,H_p},r_{b,H_p})]^{T}\tilde{w}^{p}_{i} \nonumber \\
&+\gamma^{2H_p}[\tilde{w}^{p}_{i}]^{T}\Phi^{\hat{\mu}^{p}_{i},\hat{\mu}^{p}_{-i}}_{i}(x_{b,H_p},r_{b,H_p})\nonumber \\
& \cdot [\Phi^{\hat{\mu}^{p}_{i},\hat{\mu}^{p}_{-i}}_{i}(x_{b,H_p},r_{b,H_p})]^{T}\tilde{w}^{p}_{i}
\end{align*}
and obviously $\lim_{K \to \infty} E_3 =0$. Based on \eqref{eq:25},
\begin{equation}
\label{eq:46}
\sum_{b=1}^{B}[\tilde{w}^{p+1}_{i}]^{T}\Psi_{i,b}[\Psi_{i,b}]^{T}\tilde{w}^{p+1}_{i} \geq \delta B ||\tilde{w}^{p+1}_{i}||_{2}^{2}.
\end{equation}
From \eqref{eq:45} and \eqref{eq:46},
\begin{align*}
||\tilde{w}^{p+1}_{i}||_{2}^{2}\leq & \frac{1}{\delta B}\sum_{b=1}^{B}\bigg[ [\bar{\epsilon}^{p+1}_{i}(x_{b},a_{ib},a_{-ib})\\
&- \epsilon^{p+1}_{i}(x_{b},a_{ib},a_{-ib})]^{2} + E_{3}\bigg].
\end{align*}
Note that $\hat{w}^{p+1}_{i}$ is computed with the least squares scheme \eqref{eq:24}, which minimizes \eqref{eq:23}. Then $\sum_{b=1}^{B} (\epsilon^{p+1}_{i,b})^{2} \leq \sum_{b=1}^{B} (\bar{\epsilon}^{p+1}_{i,b})^{2}$, i.e. $\sum_{b=1}^{B} |\epsilon^{p+1}_{i,b}| \leq \sum_{b=1}^{B} |\bar{\epsilon}^{p+1}_{i,b}|$. Therefore, we get
\begin{align}
\label{eq:47}
||\tilde{w}^{p+1}_{i}||_{2}^{2}\leq & \frac{1}{\delta B}\sum_{b=1}^{B}\bigg[[\bar{\epsilon}^{p+1}_{i}(x_{b},r_b,a_{ib},a_{-ib}) \nonumber \\
&- \epsilon^{p+1}_{i}(x_{b},r_b,a_{ib},a_{-ib})]^{2} + E_{3}\bigg] \nonumber \\
\leq & \frac{1}{\delta B}\sum_{b=1}^{B}\bigg[\big[|\bar{\epsilon}^{p+1}_{i}(x_{b},r_b,a_{ib},a_{-ib})| \nonumber \\
&+ |\epsilon^{p+1}_{i}(x_{b},r_b,a_{ib},a_{-ib})|\big]^{2} + E_{3}\bigg] \nonumber \\
\leq & \frac{1}{\delta B} \sum_{b=1}^{B} [ 4\cdot (\bar{\epsilon}^{p+1}_{i,max})^{2} + E_{3}] \nonumber \\
=& \frac{4}{\delta} (\epsilon^{p+1}_{i,max})^{2} + \frac{1}{\delta B} \sum_{b=1}^{B} E_{3} ,
\end{align}
where $\epsilon^{p+1}_{i,max} = \max_{b}|\bar{\epsilon}^{p+1}_{i}(x_b,r_b,a_{ib},a_{-ib})|$. Then we get $\lim_{K \to \infty} ||\tilde{w}^{p+1}_{i}||_{2}^{2}\leq 0$, i.e., $\lim_{K \to \infty} ||\tilde{w}^{p+1}_{i}||_{2}=0$. Hence, $\hat{Q}^{p+1}_{i}(x,r,a_{i},a_{-i}) - \bar{Q}^{p+1}_{i}(x,r,a_{i},a_{-i}) = [\Phi_{i}(x,r,a_{i},a_{-i})]^{T}\tilde{w}^{p+1}_{i}- \bar{e}^{p+1}_{i}(x,r,a_{i},a_{-i})$. Therefore,
\begin{equation}
\label{eq:48}
\lim_{K \to \infty} \hat{Q}^{p+1}_{i}(x,r,a_{i},a_{-i}) =\lim_{K \to \infty} \bar{Q}^{p+1}_{i}(x,r,a_{i},a_{-i}).
\end{equation}
We conclude the proof with induction. For $p=0$, we have $\bar{Q}^{0}_{i}(x,r,a_{i},a_{-i}) =\hat{Q}^{0}_{i}(x,r,a_{i},a_{-i})=Q^{0}_{i}(x,r,a_{i},a_{-i})$ and $\hat{\mu}^{-1}_{i}(x,r) = \mu^{-1}_{i}(x,r)$, which means $\hat{\mu}^{0}_{i}(x,r)=\mu^{0}_{i}(x,r)$ for all $i$. Therefore, $\lim_{K \to \infty} \bar{Q}^{0}_{i}(x,r,a_{i},a_{-i}) ={Q}^{0}_{i}(x,r,a_{i},a_{-i})$ and $\lim_{K\to \infty} \hat{\mu}^{0}_{i}(x,r) =\mu^{0}_{i}(x,r)$. Assume then that  $\lim_{K \to \infty} \bar{Q}^{p}_{i}(x,r,a_{i},a_{-i}) ={Q}^{p}_{i}(x,r,a_{i},a_{-i})$ and $\lim_{K\to \infty} \hat{\mu}^{p}_{i}(x,r) =\mu^{p}_{i}(x,r)$ for all $i$. Then,
\begin{align}
\label{eq:49}
&\lim_{K \to \infty} \bar{Q}^{p+1}_{i}(x,r,a_{i},a_{-i}) = l_{i}(x,r,a_{i},a_{-i}) \nonumber \\
&+ \lim_{K \to \infty} \sum_{m=1}^{H_{p}-1} \gamma^{m}l^{\hat{\mu}^{p}_{i},\hat{\mu}^{p}_{-i}}_{i}(x_{m},r_m)\nonumber\\
&+\lim_{K \to \infty} \gamma^{H_{p}} \bar{Q}^{p,\hat{\mu}^{p}_{i},\hat{\mu}^{p}_{-i}}_{i}(x_{H_p},r_{H_p}) \nonumber \\
&= l_{i}(x,r,a_{i},a_{-i}) + \sum_{m=1}^{H_{p}-1} \gamma^{m}l^{\mu^{p}_{i},\mu^{p}_{-i}}_{i}(x_{m},r_{m}) \nonumber \\
&+ \gamma^{H_{p}} Q^{p,\mu^{p}_{i},\mu^{p}_{-i}}_{i}(x_{H_p},r_{H_p}) \nonumber \\
&= Q^{p+1}_{i}(x,r,a_{i},a_{-i}).
\end{align}
From \eqref{eq:48} and \eqref{eq:49}, it can be finally concluded that $\lim_{K \to \infty} \hat{Q}^{p+1}_{i}(x,r,a_{i},a_{-i})= Q^{p+1}_{i}(x,r,a_{i},a_{-i})$, which implies $\lim_{K\to \infty} \hat{\mu}^{p+1}_{i}(x,r) =\mu^{p+1}_{i}(x,r)$ for all $i$. Based on Theorem $2$, it then holds that $\lim_{p,K \to \infty} \hat{Q}^{p}_{i}(x,r,a_{i},a_{-i}) = Q^{\star}_{i}(x,r,a_{i},a_{-i})$ and $\lim_{p,K \to \infty} \hat{\mu}^{p}_{i}(x,r) =\mu^{\star}_{i}(x,r)$ which satisfy \eqref{eq:11} for all $i$.

\begin{IEEEbiography}[{\includegraphics[width=1in,height=1.25in,clip,keepaspectratio]{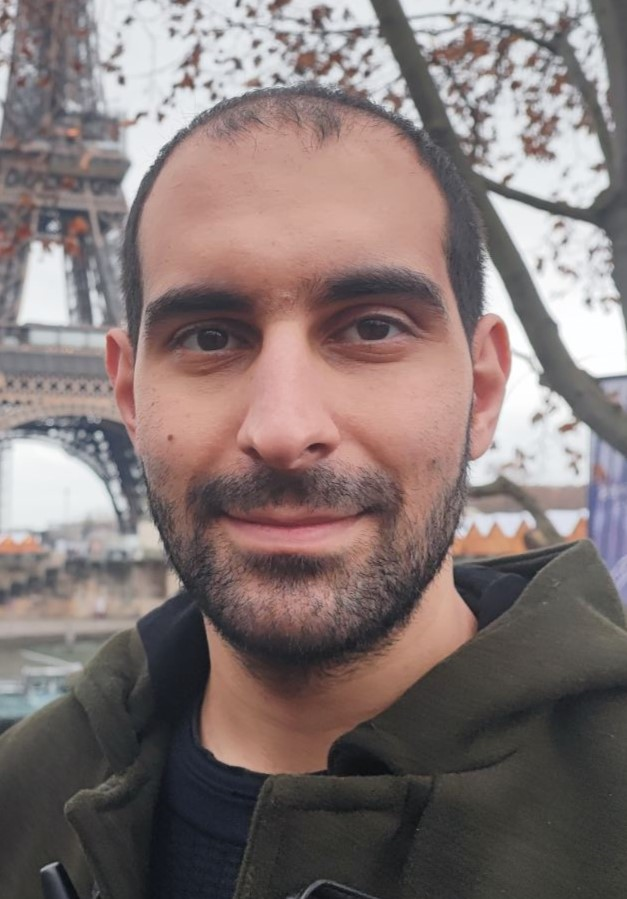}}]{Alexandros Tanzanakis}
received a 5-year Diploma (equivalent to a Master's degree) in Electrical and Computer Engineering with the highest honours from the Technical University of Crete, Greece in 2016, and a PhD in Information Technology and Electrical Engineering from ETH Zurich, Switzerland in 2023. His research interests include advanced topics in learning-based control, data-driven reinforcement learning, game theory and multiagent systems, as well as intelligent biomedical control with emphasis on the design of novel, fully-automated, personalized, closed-loop drug delivery systems.
\end{IEEEbiography}

\begin{IEEEbiography}[{\includegraphics[width=1in,height=1.25in,clip,keepaspectratio]{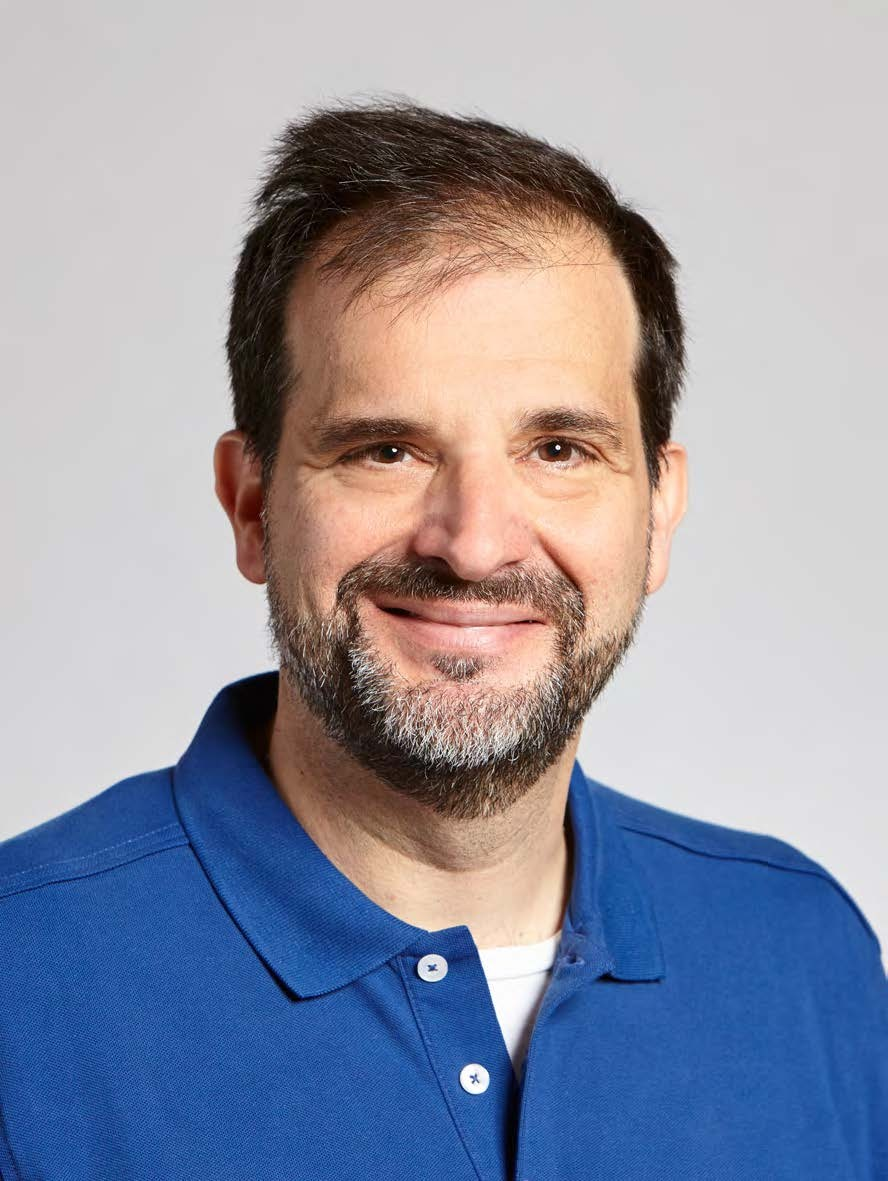}}]{John Lygeros}
received a B.Eng. degree in 1990 and an M.Sc. degree in 1991 from Imperial College, London, U.K. and a Ph.D. degree in 1996 at the University of California, Berkeley. After research appointments at M.I.T., U.C. Berkeley and SRI International, he joined the University of Cambridge in 2000 as a University Lecturer. Between March 2003 and July 2006 he was an Assistant Professor at the Department of Electrical and Computer Engineering, University of Patras, Greece. In July 2006 he joined the Automatic Control Laboratory at ETH Zurich where he is currently serving as the Professor for Computation and Control and the Head of the laboratory. His research interests include modelling, analysis, and control of large-scale systems, with applications to biochemical networks, energy systems, transportation, and industrial processes. John Lygeros is a Fellow of IEEE, and a member of IET and the Technical Chamber of Greece. Since 2013 he is serving as the Vice-President Finances and a Council Member of the International Federation of Automatic Control and since 2020 as the Director of the National Center of Competence in Research "Dependable Ubiquitous Automation" (NCCR Automation).
\end{IEEEbiography}

\end{document}